\documentclass[journal,twoside,final]{IEEEtran}
\usepackage{amsfonts,amsmath,amssymb,algorithm,algorithmic,bm}
\usepackage{graphicx,epsfig,cite}
\usepackage{color}
\usepackage{subfigure}
\usepackage[percent]{overpic}
\usepackage{multirow}
\usepackage{autobreak}
\usepackage{hyperref}
\usepackage{amsthm}
\hypersetup{hidelinks}

\newtheorem{prop}{Proposition}

\def \beqi{\begin{IEEEeqnarray}{rcl}\IEEEyesnumber}
\def \eeqi{\end{IEEEeqnarray}}

\def \bmat{\begin{bmatrix}}
\def \emat{\end{bmatrix}}

\allowdisplaybreaks[3]
\begin{document}

% \title{Joint Optimization for User Association, Resource Allocation and Beamforming in SAGIN: A Symbiotic Communication Perspective}

\title{Towards Symbiotic SAGIN Through Inter-operator Resource and Service Sharing: Joint Orchestration of User Association and Radio Resources}
% : User Association, Resource Allocation and Beamforming}

\author{Shizhao He, Jungang Ge, Ying-Chang Liang, {\it Fellow, IEEE}, and Dusit Niyato, {\it Fellow, IEEE}\vspace{-2em}
% $^\dag$University of Electronic Science and Technology of China (UESTC), Chengdu, P. R. China\\
% $^\ddag$Nanyang Technological University (NTU), Singapore\\
% Email: heshizhao@std.uestc.edu.cn, gejungang@std.uestc.edu.cn, and liangyc@ieee.org

\thanks{This work has been submitted to the IEEE for possible publication. Copyright may be transferred without notice, after which  this version may no longer be accessible.}
\thanks{S. He, J. Ge, and Y.-C. Liang are with the National Key Laboratory of Wireless Communications, and the Center for Intelligent Networking and Communications (CINC), University of Electronic Science and Technology of China (UESTC), Chengdu 611731, China (e-mail: {heshizhao@std.uestc.edu.cn; gejungang@std.uestc.edu.cn; liangyc@ieee.org}).
}
\thanks{Dusit Niyato is with the School of Computer Science and Engineering, Nanyang Technological University (NTU), Singapore (e-mail:dniyato@ntu.edu.sg)}}

\maketitle
\IEEEpubidadjcol

\begin{abstract}

The space-air-ground integrated network (SAGIN) is a pivotal architecture to support ubiquitous connectivity in the upcoming $6$G era.
Inter-operator resource and service sharing is a promising way to realize such a huge network, utilizing resources efficiently and reducing construction costs.
Given the rationality of operators, the configuration of resources and services in SAGIN should focus on both the overall system performance and individual benefits of operators.
Motivated by emerging symbiotic communication facilitating mutual benefits across different radio systems, we investigate the resource and service sharing in SAGIN from a symbiotic communication perspective in this paper.
In particular, we consider a SAGIN consisting of a ground network operator (GNO) and a satellite network operator (SNO). Specifically, we aim to maximize the weighted sum rate (WSR) of the whole SAGIN by jointly optimizing the user association, resource allocation, and beamforming. Besides, we introduce a sharing coefficient to characterize the revenue of operators. Operators may suffer revenue loss when only focusing on maximizing the WSR. 
In pursuit of mutual benefits, we propose a mutual benefit constraint (MBC) to ensure that each operator obtains revenue gains.
Then, we develop a centralized algorithm based on the successive convex approximation (SCA) method. Considering that the centralized algorithm is difficult to implement, we propose a distributed algorithm based on Lagrangian dual decomposition and the consensus alternating direction method of multipliers (ADMM). Finally, we provide extensive numerical simulations to demonstrate the effectiveness of the two proposed algorithms, and the distributed optimization algorithm can approach the performance of the centralized one. The results also reveal that the proposed MBCs can enable operators to achieve mutual benefits and realize a symbiotic resource and service sharing paradigm.
% To prompt network operators to participate in constructing SAGIN,  guaranteeing each operator's revenue, namely, achieving mutual benefits among operators is important.
% Noting that the recently proposed symbiotic communication can achieve mutual benefits among different radio systems, in this paper, we investigate the resource and service sharing in SAGIN from a symbiotic communication perspective.
% Particularly, we aim to maximize the weighted sum rate (WSR) of the whole SAGIN through jointly optimizing the user association, resource allocation, and beamforming design. Besides, we propose a mutual benefit constraint (MBC) to guarantee the revenue of each operator, and we introduce a sharing coefficient to characterize the revenue of operators.
% Then, we develop a centralized algorithm based on
% the successive convex approximation (SCA) method for the WSR maximization problem. Since the centralized algorithm is difficult to implement in practice, a distributed algorithm based on Lagrangian dual decomposition and the consensus alternating direction method of multipliers (ADMM) is also proposed for the optimization problem.
% Simulation results demonstrate the effectiveness of proposed algorithms, and the distributed algorithm can achieve a close performance to the centralized algorithm.
% Besides, by comparing the revenue of each operator in different cases, it can be observed that the proposed MBCs enable the operators to achieve mutual benefits and realize symbiotic communications in the SAGIN.


\end{abstract}

\vspace{-1em}
%\begin{IEEEkeywords}
%Spectrum sharing, Service sharing, Symbiotic communication, User association.
%\end{IEEEkeywords}

\section{Introduction} \label{sec:intro}
Wireless communication has made outstanding achievements in the past decades and stepped into the $5$G era.
Although $5$G can realize much better performance than former communication systems, including higher peak data rate, and ultra reliable and low latency, it still cannot satisfy some requirements in $6$G, such as the ubiquitous connectivity \cite{imt2030}.
In fact, there are still $2.7$ billion people who cannot access the Internet, and about $6\%$ of the rural area worldwide lacks mobile network coverage \cite{itu2022}. The principal reasons are the limited coverage of ground networks and the high expenses of deploying networks in rural areas.
Satellite networks have emerged as an important complementary technology to enhance coverage of ground networks \cite{liu2018space}, which can tackle the above issues. Compared with ground networks, satellite networks can provide seamless coverage to users without deploying costly fiber optic backhaul. 
% Fortunately, the rapid development of satellite technology boosts the development of satellite communication, which can effectively provide communication service to the remote areas with low cost, such as the Starlink.
Therefore, the space-air-ground integrated network (SAGIN) is a vital technology to realize ubiquitous connectivity in the $6$G era.

Recently, satellite networks have rapidly developed and attracted considerable attention from academia to industry. For instance, SpaceX has launched over $4,000$ low earth orbit (LEO) satellites to construct Starlink \cite{pachler2021updated}. 
To realize ubiquitous connectivity, addressing the severe path loss caused by the long propagation distance is necessary.
In Starlink, a satellite terminal (ST) comprising a satellite dish and an access point is exploited to deal with this problem. The satellite dish provides high antenna gain to compensate for the path loss, and the access point offers access service to users. In \cite{di2019ultra}, a satellite access network for $5$G and beyond is proposed based on this kind of ST.
Satellite networks are also expected to realize high throughput satellite communication. To this end, high frequency is adopted in satellite communication, such as Ka band \cite{perez2019signal}. Besides, multi-beam satellites also play an essential role in realizing high throughput satellite communication. Compared with conventional satellites, multi-beam satellites can generate multiple beams with high gain and thus provide better satellite communication \cite{joroughi2016generalized,wang2021resource}. 
Moreover, the shadowed-Rician (SR) fading is widely adopted to model the satellite-ground channels, based on which many works analyze the performance of satellite networks \cite{abdi2003new,na2021performance,kim2023downlink}.
Resource allocation in satellite networks is also widely studied \cite{gu2021dynamic,hu2018deep}. Additionally, diverse multiple access techniques are adopted in satellite networks to serve users \cite{wang2020noma,li2022user}, such as the time division multiple access (TDMA) manner.

The rapid growth of satellite networks paves the way for realizing SAGIN, which has been widely studied in recent years. Since SAGIN is a large-scale network, it is almost impossible for an individual network operator to construct it. 
% Besides, considering that different networks in SAGIN usually have complementary features, they can achieve better performance via collaboration.
Therefore, inter-operator sharing\footnote{For ease of notation, ``inter-operator sharing'' is synonymous with ``inter-operator resource and service sharing'' in this paper.} is a promising technology for constructing SAGIN.
% Since it is almost impossible for an individual network operator to construct such a large-scale network, many satellite network operators (SNOs) and ground network operators (GNOs) are trying to collaborate. 
In fact, the famous ground network operator (GNO) T-mobile announces that it aims to provide seamless coverage to users by cooperating with satellite network operator (SNO) SpaceX \cite{tmobile2022}. 
% Besides, considering that different networks in SAGIN usually have complementary features, they can achieve better performance via inter-operator sharing\footnote{For ease of notation, ``inter-operator sharing'' is synonymous with ``inter-operator resource and service sharing'' in this paper.}. 
Optimizing the configuration of resources and services is an essential problem to fully exploit the advantages of SAGIN with inter-operator sharing.
% Since an SAGIN incorporates diversity resources and services from different networks,
Considering the limited spectrum, many works study the spectrum sharing problem in SAGIN \cite{wang2021joint,liang2021realizing,lyu2021service}.
% In \cite{rahman2021game}, the authors exploit the stochastic geometry to derive the rate expression for ground networks when satellite and ground networks share the $28$ GHz band. In \cite{lyu2021service}, the spectrum belongs to different networks in SAGIN is uniformly sliced for isolated vehicle services to avoid interference. To maximize the system revenue incorporating the throughput and UAV dispatching cost, both spectrum slicing and UAV dispatching are optimized. Besides, the spectrum sharing in the downlink of SAGIN is studied in \cite{wang2021joint}. The authors aim to maximize the sum rate of SAGIN while guaranteeing the quality-of-service of ground users by optimizing beamforming of satellites and BSs.
In \cite{liang2021realizing}, catering to the heterogeneity and high dynamics of SAGIN, the authors propose an intelligent spectrum management framework empowered by artificial intelligence and software defined network (SDN). 
In \cite{zhang2022resource}, the spectrum allocation is optimized to maximize the sum rate while reducing the interference received at satellite users.
% the uplink of satellite networks and that of ground networks operate in the same spectrum band. To maximize the throughput of ground networks and minimize the interference received by satellites, the channel assignment and power allocation are optimized.
% Apart from spectrum sharing, service sharing in SAGIN is also essential to fully exploit the complementary advantages of different networks.
On the other hand, service sharing is another critical topic in studying SAGIN. Many recent scientific literatures demonstrate that satellite networks are a promising supplement to ground networks in providing diversity services to remote areas, such as backhaul connectivity \cite{li2020maritime,alsharoa2020improvement,liu2022joint} and task offloading \cite{di2018ultra,han2023two}. In \cite{liu2022joint}, a satellite assists ground networks in providing backhaul links to users in remote areas, and both user association and beamforming are optimized to maximize the sum rate.
In \cite{han2023two}, satellite networks support task offloading of remote Internet-of-Things (IoT) devices. An optimization problem about user association and spectrum allocation is formulated to minimize the  task offloading delay. The above works only focus on the overall performance of SAGIN. However, motivating different operators to share resources and services is also a vital problem. As the operators are inherently competitors with conflicting interests, they will be reluctant to construct SAGIN if resource and service sharing cannot guarantee their revenue, i.e., they cannot achieve mutual benefits.
In \cite{deng2020ultra}, a pricing mechanism is proposed to ensure each 
operator's revenue. However, it cannot fully leverage the resources in SAGIN to achieve a high overall performance.
Therefore, it is vital to simultaneously achieve mutual benefits among operators and maximize the overall performance of SAGIN. 

The recently emerging symbiotic communication provides a promising paradigm to tackle the above problem, which aims to optimize the collective objectives of diverse radio systems and realize mutual benefits among them via sharing resources and services \cite{liang2022symbitoic}.
% symbiotic communication leverages the analogy to the natural ecosystem in biology and regards the electromagnetic space as a radio ecosystem \cite{liang2022symbitoic}. It aims to establish a symbiotic relationship among different radio systems, namely, optimizing the collective objectives of varying radio systems and realizing mutual benefits among them via the proper configuration of resources and services.
Symbiotic communication has been widely studied in backscatter communications \cite{long2020symbiotic,zhou2023assistance,zhang2024channel,wang2024multi,chen2023transmission}. In \cite{long2020symbiotic}, the authors investigate the symbiotic communication in a passive IoT system, consisting of an active transmission and a backscatter transmission. Specifically, the active transmission shares spectrum and energy resources with the backscatter transmission, and in return the backscatter transmission provides a beneficial multi-path to the active transmission via reflecting the incident signal. As a result, the active and backscatter transmission achieve mutual benefits in this system, realizing higher sum rate.
% the backscatter transmission obtains transmission opportunity, and the active transmission has the multi-path diversity gain, achieving mutual benefits and higher sum rate.
Besides, the symbiotic communication is realized in the coexistence of a cellular network and WiFi system in \cite{tan2019qos}. Particularly, the WiFi system shares spectrum with the cellular network, and the cellular network adjusts transmission parameters to guarantee the performance of WiFi system.
% is realized by optimizing the user association scheme and allocation of resources.
Resource allocation in symbiotic communication is also studied.
In \cite{he2023user,he2023joint}, both user association and beamforming are optimized to realize the symbiotic communication in multi-operator cellular networks. Symbiotic communication can also be a prominent solution to construct SAGIN.
% Since many different networks coexist in SAGIN, symbiotic communication can also be a prominent solution to the resource and service configuration problem in SAGIN.
In \cite{cheng2022blockchain}, the authors propose a blockchain and deep learning based framework to optimize resource allocation in SAGIN with the guidance of symbiotic communication.
However, it lacks an explicit formulation to strictly guarantee individual benefits after sharing resources and services.

% Motivated by the above works, in this paper, we investigate joint user association, resource allocation, and beamforming design in SAGIN from a symbiotic communication perspective. 
Motivated by the above works, in this paper, we investigate joint user association, resource allocation, and beamforming design in SAGIN from a symbiotic communication perspective. Specifically, the considered SAGIN system consists of a GNO and SNO engaging in inter-operator sharing, i.e., sharing spectrum and allowing users to access arbitrary networks. The GNO has base stations (BSs) connected to the core network by fiber. On the other hand, the SNO has STs connected to the core network via a multi-beam LEO satellite, and this LEO satellite serves STs in the TDMA manner. Given the high cost of deploying fiber in remote areas, the number of BSs is fewer than that of STs.
Our objective is to maximize the weighted sum rate (WSR) of all users while achieving mutual benefits among operators in terms of revenue. To characterize revenue of operators, we introduce a sharing coefficient that reflects the impact of inter-operator sharing.
% which influences the revenue acquired from serving users.
The revenue may experience losses when only pursuing the maximization of the WSR.
% With a small sharing coefficient, the operator may suffer revenue loss during service sharing.
To achieve mutual benefits, we propose a mutual benefit constraint (MBC) to guarantee that each operator can obtain higher revenue with inter-operator sharing.
% Besides, we introduce a sharing coefficient to characterize the revenue of operators. Then, we propose a mutual benefit constraint (MBC) to guarantee that each operator can obtain higher revenue with inter-operator sharing, achieving mutual benefits.
% To achieve mutual benefits among SNO and GNO, we propose a mutual benefit constraint (MBC) to guarantee the revenue of each operator, and the weighted sum rate (WSR) of the SAGIN system is optimized. 
% Besides, considering the additional cost incurred by service sharing, 
% considered by jointly optimizing the power and time allocation of the LEO, user association scheme and beamforming design. 
% Besides, we introduce a sharing coefficient to model the impact of the additional cost incurred by inter-operator sharing on the revenue of an operator.
The main contributions of this work can be summarized as follows:
\begin{itemize}
  \item We investigate a symbiotic resource and service sharing paradigm in the SAGIN, which aims to maximize the overall system performance and realize mutual benefits between the GNO and SNO. We formulate an optimization problem regarding user association, resource allocation, and beamforming design to maximize the WSR, subject to the MBCs and STs' backhaul capacity constraints.
  \item We introduce a sharing coefficient to characterize the revenue of each network operator, based on which the MBC is formulated to ensure each operator can obtain revenue gains after sharing spectrum and services.
  \item We propose a centralized algorithm based on successive convex approximation (SCA) for the formulated WSR maximization problem. Concerning the large scale of the SAGIN, it is challenging to implement the centralized algorithm. Then, we develop a distributed algorithm, which consists of the Lagrangian dual decomposition based algorithm for user association and consensus alternating direction method of multipliers (ADMM) based algorithm for beamforming design and resource allocation.
  \item Simulation results show that both the proposed algorithms can outperform the considered benchmarks, and the distributed algorithm approaches a performance close to that of the centralized algorithm. Moreover, in comparison to operators in the non-symbiotic case, each operator in the symbiotic case attains higher revenue after sharing spectrum and services. This implies that operators achieve mutual benefits and realize a symbiotic resource and service sharing paradigm in the SAGIN. 
% the results demonstrate that the MBCs can help network operators achieve mutual benefits and realize a symbiotic resource and service sharing paradigm in the SAGIN.
  % MBCs are necessary to realize symbiotic communication in SAGIN by comparing revenue of operators obtained with MBCs and that without MBCs.
\end{itemize}

The remainder of this paper is organized as follows. Section \ref{sec:sysmodel} illustrates the system model. In Section \ref{sec:opt_problem}, we formulate the WSR maximization problem for the considered SAGIN system. In Section \ref{sec:cen_alg}, an SCA-based centralized algorithm and finding initial point algorithm are proposed for the WSR maximization problem. In Section \ref{sec:dis_alg}, a distributed algorithm based on consensus ADMM and Lagrangian dual decomposition methods is developed. In Section \ref{sec:sim_result}, the extensive simulation results are presented.
% and the revenue of SNO and GNO obtained with and without MBCs are analyzed. 
Finally, Section \ref{sec:conclu} concludes this paper.

Notations: The notations in this paper are listed as follows. The scalars, column vectors and matrices are represented by lowercase, bold lowercase and uppercase symbols (e.g., $x$, $\mathbf{x}$ and $\mathbf{X}$), respectively. $\left\vert \mathcal{A} \right\vert$ denotes the cardinality of set $\mathcal{A}$. $\mathbb{C}$ denotes the set of complex numbers. Main notations are summarized in Table \ref{tab:notation}.
%  Besides, $\mathbb{R}$, $\mathbb{R}_+$ and $\mathbb{R}_{++}$ respectively denote the sets of real numbers, nonnegative real numbers and positive real numbers.

\begin{table*}[htbp]
  \centering
  \caption{Main Notations}\label{tab:notation}
  \renewcommand\arraystretch{1.3}
  \resizebox{0.6\textwidth }{!}{
  \begin{tabular}{l| l}
    \hline
    \textbf{Notation} & \textbf{Meaning} \\ \hline
    % \hline 
    $ G, S, \mathcal{O}$ & GNO, SNO, set of network operators \\ 
    % \hline
    $\mathcal{B}, \mathcal{B}_G, \mathcal{B}_S$ & Set of BSs and STs, set of BSs, set of STs  \\ 
    % \hline
    $\mathcal{U}, \mathcal{U}_G, \mathcal{U}_S$ & Set of all users, set of users from GNO, set of users from SNO \\ 
    % \hline
    $N_B, N_B^\prime$ & The number of BSs, the number of STs \\ 
    % \hline
    $N_U, N_U^\prime$ & The number of users from GNO, the number of users from SNO \\
    % \hline
    $N_L, N_t$ & The number of feeds of the LEO satellite, the number of transmit antennas of BSs/STs \\
    % \hline
    $P_i, P_{Sat}$ & Maximum transmit power of BS/ST $i$, maximum transmit power of the LEO satellite \\ 
    % \hline
    $G_T(\cdot), G_R$ & Transmit antenna gain of the $\ell$-th spot beam, receive antenna gain of STs\\ 
    % \hline
    % $\zeta_{\ell,i}$ & Off-boresight angle between the the $\ell$-th spot beam and ST $i$ \\ \hline
    $f_i$ & Downlink channel between the LEO satellite and ST $i$ \\ 
    % \hline
    $\mathbf{h}_{i,k}^n$ & Downlink channel between BS/ST $i$ and user $k$ over band $n$ \\ 
    % \hline
    % $B_{Ka}, B_C$ & Bandwidth of Ka band, bandwidth of C band  \\ \hline
    $p_{\ell}, t_{\ell,i}$ & Power of $\ell$-th spot beam, time allocated to ST $i$ by $\ell$-th spot beam \\ 
    % \hline
    $x_{i,k}$ & User association variable for BS/ST $i$ and user $k$ \\ 
    % \hline
    $\mathbf{w}_{i,k}^n$ & Transmit beamforming that BS/ST $i$ assigns to user $k$ over band $n$  \\ 
    % \hline
    % $\tilde{I}_{i,k}^n, I_{i,k}^n$ & Intra-operator interference, inter-operator interference \\ \hline
    % $\sigma^2_s, \sigma^2_t$ & Noise power receive at STs, noise power received at users \\ \hline
    % $R_{i,k}, C_{\ell,i}$ & Achievable rate of user $k$ associated with BS/ST $i$, backhaul capacity for ST $i$ served by $\ell$-th spot beam \\ \hline
    $\delta_z$ &  Sharing coefficient of operator $z$ \\ 
    % \hline
    $U_z^0, U_z$ & Revenue without inter-operator sharing, revenue with inter-operator sharing \\ \hline
    % $\gamma_{i,k}^n, \Gamma_{i,k}^n, \varphi_{i,k}^n$ & SINR, lower bound of SINR, upper bound of SINR for user $k$ associated with BS/ST $i$ over band $n$ \\ \hline
    % $\beta_{i,k}^n, \rho_{i,k}^n$ & Upper bound and lower bound of interference received at user $k$ associated with BS/ST $i$ over band $n$, respectively \\ \hline
    % $\dot{\Gamma_{i,k}^{n}}, \dot{\theta_{i,k}^{n}}, \dot{\psi_{i,k}^{n}}$ & Global variables in Algorithm. \ref{Alg5} SINR, upper bound, and lower bound of inter-operator interference, respectively \\ \hline
    % $\Gamma_{i,k}^{n,(z)}, \theta_{i,k}^{n,(z)}, \psi_{i,k}^{n,(z)}$ & Local variables optimized by operator $z$ in Algorithm. \ref{Alg5}  \\ \hline
    % $\lambda_z, \nu_{\Gamma}, \nu_{\theta}, \nu_{\psi}$ & Dual variables introduced to Lagrangian \\ \hline
    % $\eta, \varrho, \xi$ & Step size, penalty parameters. \\ \hline
  \end{tabular}
  }
  \vspace{-1em}
\end{table*}

\section{System Model}\label{sec:sysmodel}

\begin{figure}[tp]
  % \vspace{-1em}
\begin{center}
\epsfxsize=0.4\textwidth \leavevmode
\epsffile{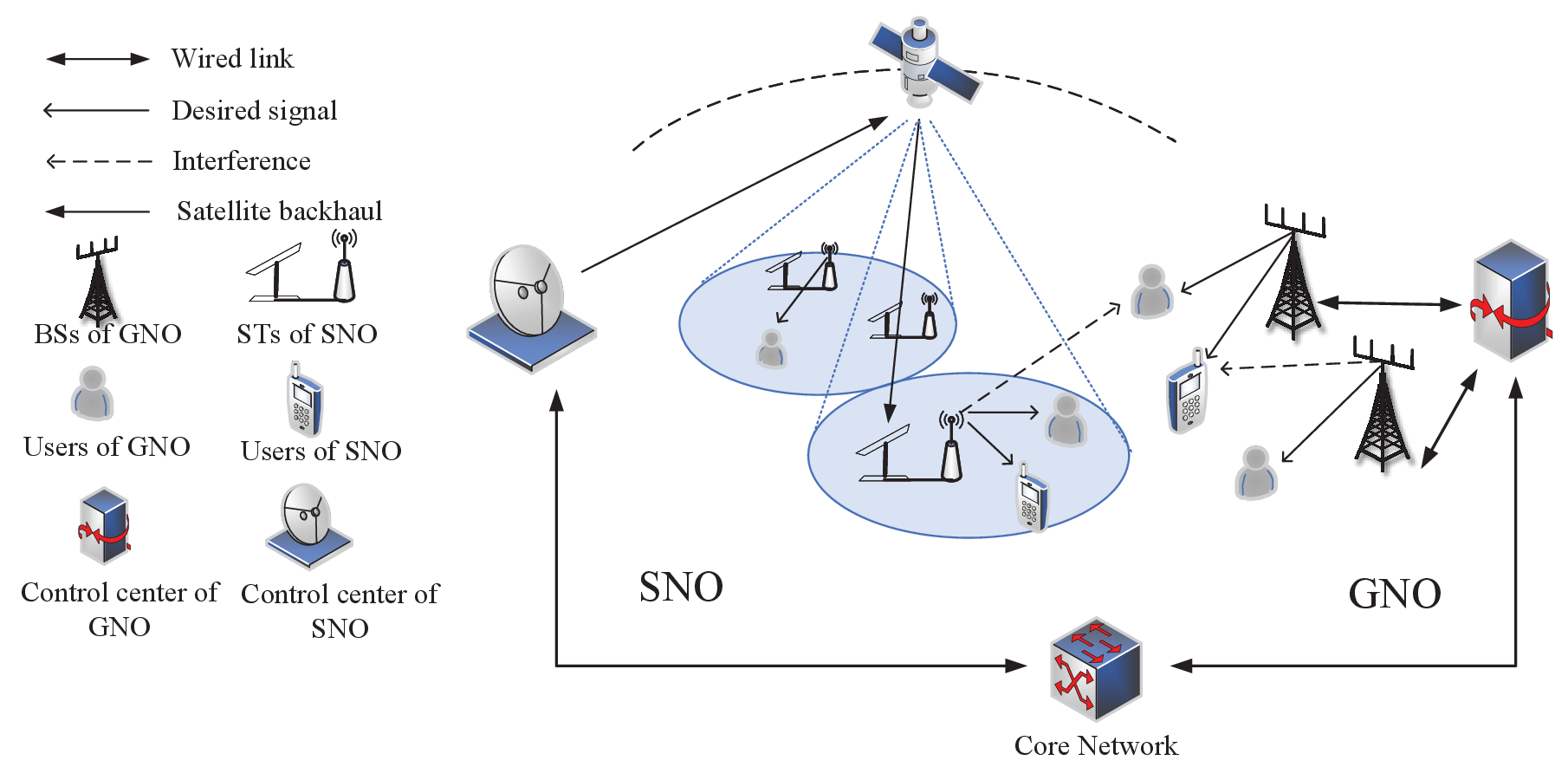}
\caption{The considered SAGIN with inter-operator sharing.}\label{fig:sys_model}
\end{center}
\vspace{-1em}
\end{figure}

\subsection{SAGIN Network}
As shown in Fig. \ref{fig:sys_model}, we consider the downlink of a SAGIN system, which consists of a GNO and SNO, and we define $\mathcal{O}=\{G,S\}$ as the set of operators. Specifically, the GNO has $N_U$ users and deploys $N_B$ BSs, directly connected to the core network via fiber that can provide large backhaul capacity. On the other hand, the SNO has $N_U^\prime$ subscribed users and deploys $N_B^\prime$ STs to serve users in remote areas, and it exploits a multi-beam LEO satellite to provide backhaul links to STs over the Ka band. Since the payload of the satellite is constrained, the backhaul capacity it provides to STs is limited. Besides, both operators have a dedicated C band for serving users.
We denote $\mathcal{B}_G=\{1,2,\ldots,N_B\}$ and $\mathcal{B}_S=\{N_B+1,N_B+2,\ldots,N_B+N_B^\prime\}$ as the set of all BSs and STs, respectively. Denote the set of subscribed users of the GNO as $\mathcal{U}_G=\{1,2,\ldots,N_U\}$ and that of the SNO as $\mathcal{U}_S=\{N_U+1,N_U+2,\ldots,N_U+N_U^\prime\}$. Then, we denote $\mathcal{B}\triangleq \cup_{n\in\mathcal{O}} \mathcal{B}_n$ and $\mathcal{U}\triangleq \cup_{n\in\mathcal{O}} \mathcal{U}_n$. We consider that all BSs and STs are equipped with $N_t$ transmit antennas while all users are equipped with a single antenna.
In the considered SAGIN, inter-operator sharing is adopted, where all BSs and STs share their C bands to serve users, and the users can access arbitrary BSs or STs for service. Furthermore, we assume that each operator has a control center to
collect necessary information (e.g., channel state information (CSI), and user locations.).
% in which $1\leq i \leq N_B$ represents the BSs while $N_B+1\leq i\leq N_B+N_B^\prime$ indicates the STs. Denote all users as $\mathcal{text}$
% each operator $m$ deploys a set $\mathcal{B}_m$ of BSs equipped with $N_t$ transmit antennas and has a set $\mathcal{U}_m$ of users equipped with a single antenna. Then, we denote $\mathcal{B}\triangleq \cup_{m\in\mathcal{M}} \mathcal{B}_m$ and $\mathcal{U}\triangleq \cup_{m\in\mathcal{M}} \mathcal{U}_m$ as all BSs and users in the considered multi-operator network, respectively. 

% We assume each operator $n$ has an exclusive C band of flat fading to serve its own users, and the downlink channel between BS/ST $i\in\mathcal{B}_m$ and user $k\in\mathcal{U}_n$ in the spectrum band of operator $c\in\mathcal{M}$ is given by
% \begin{equation}\label{eq:channel_gain}
%   \mathbf{h}_{i,k}^{c} = \sqrt{d_{i,k}^{c}}\tilde{\mathbf{h}}_{i,k}^{c},
% \end{equation}
% where $L_{i,k}^{c}$ accounts for the path-loss between BS/ST $i$ and user $k$, and $\tilde{\mathbf{h}}_{i,k}^{c} \in \mathbb{C}^{1\times N_t}$ is the small-scale fading coefficient. 

\subsection{Space-Ground Communication Model}\label{subsec:model_s2g}
We consider that the multi-beam LEO satellite is equipped with $N_L$ feeds, and each can generate a spot beam to provide backhaul links to STs. Specifically, each beam serves STs within its coverage in the TDMA manner. Besides, universal frequency reuse is adopted in space-ground communication, which means all beams operate over the Ka band. Then, we assume each ST is pointed towards the satellite, and the received signal of ST $i$ within the coverage of beam $\ell$ is
\begin{equation}\label{eq:bkhaul_signal}
  \begin{aligned}
    y_{\ell, i} &= \sqrt{p_{\ell}G_T(\zeta_{\ell,i})G_R} f_i s_\ell \\
    & + \sum_{\substack{\ell^\prime=1,\\\ell^\prime \neq \ell}}^{N_L} 
    % \sum_{j\in\mathcal{S}_{\ell^\prime}} 
    \sqrt{p_{\ell^\prime}G_T(\zeta_{\ell^\prime,i})G_R} f_i s_{\ell^\prime} + z_i,
  \end{aligned}
\end{equation}
where $p_{\ell}$ is the transmit power of $\ell$-th spot beam; 
% $\mathcal{S}_{\ell^\prime}$ denotes STs served by $\ell$-th spot beam;
$s_{\ell}$ is the symbol transmitted by $\ell$-th spot beam, which is assumed to be an independent random variable with zero mean and unit variance; and $z_{i} \sim \mathcal{CN}(0, \sigma_s^2) $ is the additive white Gaussian noise (AWGN).
$G_R$ is the receive antenna gain of STs, $G_T(\zeta_{\ell,i})$ is the transmit antenna gain of the $\ell$-th spot beam toward ST $i$, and $\zeta_{\ell,i}$ is the off-boresight angle between the center of $\ell$-th spot beam and ST $i$. Specifically, $G_T(\zeta_{\ell,i})$ is given by
\begin{equation}
  G_T(\zeta_{\ell,i}) = G_T^0 \cdot 4 \left\vert\frac{J_1(\kappa a\sin\zeta_{\ell,i})}{\kappa a\sin\zeta_{\ell,i}} \right\vert^2, \nonumber
\end{equation}
where $G_T^0$ is the maximum transmit gain, $\kappa = \frac{2\pi}{\lambda_{Ka}}$ is the wave number, $a$ is the radius of the dish antenna, and $\lambda_{Ka}$ is the carrier length of Ka band.
$f_i$ is the downlink channel between the LEO satellite and ST $i$, which is given by
\begin{equation}\label{eq:sg_channel_gain}
  f_i = \sqrt{d_i}\tilde{f}_i,
\end{equation}
where $d_i$ accounts for the path-loss between the satellite and ST $i$, while $\tilde{f}_i$ is the SR fading. The probability distribution function of SR fading gain $\vert \tilde{f}_i \vert^2$ is given by
\begin{equation}\label{eq:sr_gain}
  \begin{aligned}
  f_{\vert \tilde{f}_i \vert^2}(s;b,m,\Omega) &= \frac{1}{2b}\left(\frac{2bm}{2bm+\Omega}\right)^m\exp\left(-\frac{s}{2b}\right)\times \\
  & _1F_1\left(m,1,\frac{\Omega s}{2b(2bm+\Omega)}\right),
  \end{aligned}
\end{equation}
where $\Omega$ is the average power of a LoS component, $2b$ is the average power of the scatter components, $m$ is the Nakagami parameter, and $_1F_1(\cdot,\cdot,\cdot)$ is the confluent hypergeometric function. Based on \eqref{eq:bkhaul_signal}, the backhaul capacity for ST $i$ served by $\ell$-th spot beam is
\begin{equation}\label{eq:bkhaul_cap}
  \begin{aligned}
  C_{\ell,i} &= B_{Ka}t_{\ell,i} \\
  & \times \log_2\left(1+\frac{p_{\ell} G_T(\zeta_{\ell,i})G_R\vert f_i \vert^2}{\sum_{\ell^\prime \neq \ell}p_{\ell^\prime}G_T(\theta_{\ell^\prime,i})G_R\vert f_i \vert^2 + \sigma_s^2}\right),
  \end{aligned}
\end{equation}
where $B_{Ka}$ denotes the bandwidth of Ka band, $t_{\ell,i}$ is the fraction of time resource $\ell$-th spot beam allocates to ST $i$.

\subsection{Ground-Ground Communication Model}\label{subsec:model_g2g}
% We assume both GNO and SNO have an exclusive C band of flat fading to serve their own users, 
Universal frequency reuse is also adopted by BSs/STs to serve users, and the downlink channel between BS/ST $i$ and user $k$ over the spectrum band of operator $n\in\mathcal{O}$ is given by
\begin{equation}\label{eq:channel_gain}
  \mathbf{h}_{i,k}^{n} = \sqrt{d_{i,k}}\tilde{\mathbf{h}}_{i,k}^{n},
\end{equation}
where $d_{i,k}$ is the path-loss between BS/ST $i$ and user $k$, and $\tilde{\mathbf{h}}_{i,k}^{n} \in \mathbb{C}^{1\times N_t}$ is the small-scale fading coefficient. 
With inter-operator sharing, the users of one operator can access other operators' networks, and each BS/ST can operate over all spectrum bands. In this case, the received signal of user $k$ associated with BS/ST $i\in\mathcal{B}_m$ over the band $n$ is 
\begin{equation}\label{eq:sha_signal}
  \begin{aligned}
    y_{i,k}^n &= \mathbf{h}_{i,k}^n \mathbf{w}_{i,k}^ns_{k} \\
    &+ \sum_{\substack{j\in\mathcal{B}_m,k^\prime\in\mathcal{U}\\(j,k^\prime)\neq(i,k)}}
    % \sum_{j\in\mathcal{B}_m/\{i\}} \sum_{k^\prime\in\mathcal{U}/\{k\}} 
    \mathbf{h}_{j,k}^n \mathbf{w}_{j,k^\prime}^n s_{k^\prime}\\
    & + \sum_{j\in\mathcal{B}\setminus\mathcal{B}_m,k^\prime\in\mathcal{U}}
    % \sum_{j\in\mathcal{B}/\mathcal{B}_m} \sum_{k^\prime\in\mathcal{U}/\{k\}}
    \mathbf{h}_{j,k}^n \mathbf{w}_{j,k^\prime}^n s_{k^\prime} + z^n,
  \end{aligned}
\end{equation}
where $\mathbf{w}_{i,k}^n\in\mathbb{C}^{N_t \times 1}$ is the transmit beamforming that BS/ST $i$ assigns to user $k$ over the band $n$; $s_{k}$ denotes the symbol transmitted to user $k$; and $z^n \sim \mathcal{CN}(0, \sigma_t^2) $ is the AWGN over band $n$.
The instantaneous signal-to-interference-plus-noise ratio (SINR) of user $k$ in the band $n$ can be written as
\begin{equation}\label{eq:sha_sinr}
  \gamma_{i,k}^{n} = \frac{\left\vert \mathbf{h}_{i,k}^n \mathbf{w}_{i,k}^n \right\vert ^2}{\tilde{I}_{i,k}^n + I_{i,k}^n + \sigma_t^2},
\end{equation}
where $\tilde{I}_{i,k}^n$ denotes the intra-operator interference while $I_{i,k}^n$ denotes the inter-operator interference. Specifically, $\tilde{I}_{i,k}^n$ and $I_{i,k}^n$ are as follows:
\begin{subequations}\label{eq:sha_interfer}
  \begin{align}
    &\tilde{I}_{i,k}^n = \sum_{\substack{j\in\mathcal{B}_m,k^\prime\in\mathcal{U}\\(j,k^\prime)\neq(i,k)}} \left\vert\mathbf{h}_{j,k}^n \mathbf{w}_{j,k^\prime}^n\right\vert ^2, \label{eq:sha_intra_interfer}\\
    &I_{i,k}^n = \sum_{j\in\mathcal{B}\setminus\mathcal{B}_m,k^\prime\in\mathcal{U}} \left\vert\mathbf{h}_{j,k}^n\mathbf{w}_{j,k^\prime}^n\right\vert ^2. \label{eq:sha_inter_interfer}
  \end{align}
\end{subequations}
Thus, the achievable rate for user $k$ with inter-operator sharing is shown as
\begin{equation}\label{eq:sha_rate}
  R_{i,k} = \sum_{n\in\mathcal{O}}\log_2(1 + \gamma_{i,k}^n).
\end{equation}

% \begin{equation}\label{eq:sha_rate}
%   R_{mi,nk}(t) = \sum_{c\in\mathcal{M}}\log_2(1 + \gamma_{mi,nk}^c(t)).
% \end{equation}

\section{Problem Formulation}\label{sec:opt_problem}

As the network performance heavily depends on the configuration of resources and services, we formulate a WSR maximization (WSRM) problem through jointly optimizing user association, resource allocation, and beamforming design. Specifically, since the backhaul capacity offered by the LEO satellite is limited, the investigated WSRM problem is subject to the backhaul capacity constraints of STs. Besides, we introduce a sharing coefficient $\delta_z$ for operator $z$ to characterize its revenue. Based on the revenue, MBCs are proposed to realize mutual benefits, realizing a symbiotic resource and service sharing paradigm.
% to encourage operators to share spectrum and services, MBCs are proposed to guarantee unimpaired revenue for all operators. 
% Specifically, we formulate WSR maximization (WSRM) problems for two cases: the case with spectrum and service sharing, and the case without spectrum or service sharing. The latter one can be adopted as a benchmark to demonstrate the performance improvement of spectrum and service sharing.

\subsection{Revenue of Operators}\label{subsec:s2_sha}
While inter-operator sharing can improve the overall performance, it concurrently incurs additional costs due to inter-operator interference and serving other operators' users.
Hence, to compensate for the additional costs and incentivize operators to engage in inter-operator sharing,
a sharing coefficient $\delta_z \in [0, 1]$ is introduced for each operator $z$ to characterize its revenue.
As the revenue of an operator is usually related to the transmission rate that it can offer, we define the revenue of operator $z$ as
\begin{equation}\label{eq:revenue_op}
  % U_z = \sum_{n\in\mathcal{O}, k\in\mathcal{U}_n}U_{z,k}=\sum_{i\in\mathcal{B},k\in\mathcal{U}}x_{i,k}\alpha_{i,k}^zR_{i,k},
  U_z =\sum_{i\in\mathcal{B},k\in\mathcal{U}}x_{i,k}\alpha_{i,k}^zR_{i,k},
\end{equation}
where 
% $U_{z,k}$ denotes the revenue operator $z$ obtains when user $k$ is served by BS/ST $i$, and 
$x_{i,k}\in\{0,1\}$ is the user association variable, $x_{i,k}=1$ if user $k$ is associated with BS/ST $i$, and $x_{i,k}=0$ otherwise. Let BS/ST $i$ belong to operator $m$, and user $k$ subscribe to operator $n$, then $\alpha_{i,k}^z$ is defined as
\begin{equation}\label{eq:roaming_coefficient}
  \alpha_{i,k}^z = \begin{cases}
                      1, & \mbox{if } z=m=n,\\
                      \delta_m, & \mbox{if } z=m,m \neq n,\\
                      1-\delta_m, & \mbox{if } z=n,m \neq n,\\
                      0, & \mbox{otherwise.}
               \end{cases}
\end{equation}
From \eqref{eq:revenue_op} and \eqref{eq:roaming_coefficient}, it can be observed that the sharing coefficient $\delta_m$ can affect 
the revenue (or compensation) obtained by operator $m$ from serving other operators' users. Thus, this sharing coefficient has a significant impact on service sharing. 
For example, when $\delta_m=0$, other operators need not compensate operator $m$ for offloading their users to it, and operator $m$ 
cannot obtain any revenue. 
As a result, for higher revenue, operator $m$ prefers to serve its users, while other operators prefer to offload their users to operator $m$.

\subsection{WSRM Formulation}
In the studied SAGIN system, each user can be associated with any BSs/STs, and each BS/ST can operate in all bands. To efficiently maximize performance of the whole system while ensuring that operators achieve mutual benefits, we formulate the WSRM problem from the symbiotic communication perspective, which is written as
\begin{subequations}\label{eq:p1}
  \begin{align}
    \mathbf{P1:} & \max_{\mathbf{p,t,W,x}} \sum_{i\in\mathcal{B}}\sum_{k\in\mathcal{U}} x_{i,k}b_{i,k}R_{i,k} \label{eq:p1_obj} \\
    & \text{s.t. } U_z \geq U_z^0, \forall z \in\mathcal{O}, \label{eq:mbc} \\
    & \quad\ \
    \sum_{i \in \mathcal{B}_S} t_{\ell,i} \leq 1, \forall \ell=1,\cdots,N_L, \label{eq:p1_ta} \\
    & \quad\ \
    \sum_{k\in\mathcal{U}} B_CR_{i,k} \leq C_{\ell,i}, \forall i \in \mathcal{B}_S, \label{eq:p1_bkhaul} \\
    & \quad\ \
    \sum_{n\in\mathcal{O}}\sum_{k\in\mathcal{U}} \parallel\mathbf{w}_{i,k}^n\parallel^2 \leq P_i, \forall i \in \mathcal{B}, \label{eq:p1_bs_power} \\
    & \quad\ \
    \sum_{\ell=1}^{N_L} p_{\ell} \leq P_{Sat}, \label{eq:p1_sat_power} \\
    & \quad\ \
    \sum_{i\in\mathcal{B}} x_{i,k} \leq 1, \forall k \in \mathcal{U}, \label{eq:p1_ua} \\
    & \quad\ \
    x_{i,k} \in \{0, 1\}, \forall i\in\mathcal{B},\forall k \in \mathcal{U}, \label{eq:p1_binary}
    % & \quad\ \
  \end{align}
\end{subequations}
where $\mathbf{p}\triangleq\{p_{\ell}\}_{\ell=1,2,\cdots,N_L}$, $\mathbf{t}\triangleq\{t_{\ell,i}\}_{\ell=1,2,\cdots,N_L, i\in\mathcal{B}_S}$, $\mathbf{w}\triangleq\{\mathbf{w}_{i,k}^n\}_{n\in\mathcal{O}, i\in\mathcal{B},k\in\mathcal{U}}$, $\mathbf{x}\triangleq\{x_{i,k}\}_{\in\mathcal{B},k\in\mathcal{U}}$, $b_{i,k}$ denotes the weight for the rate of user $k$ associated with BS $i$, and $P_i$ and $P_{Sat}$ are the maximum transmit power of BS/ST $i$ and the LEO satellite, respectively. Besides, \eqref{eq:p1_ta} is the time resource allocation constraint; \eqref{eq:p1_bkhaul} is the backhaul constraint for each ST, and $B_C$ is the bandwidth of C band; \eqref{eq:p1_bs_power} is the power constraint for each BS/ST, which indicates the spectrum sharing is enabled; \eqref{eq:p1_sat_power} is the power constraint for the LEO satellite; \eqref{eq:p1_ua} means that each user can be associated with only one BS/ST, which indicates the service sharing is enabled, and \eqref{eq:p1_binary} indicates that $x_{i,k}$ is a binary variable. Further, \eqref{eq:mbc} is referred to as the MBC, which ensures that each operator can obtain benefits. Specifically, $U_z$ is the revenue operator $z$ can obtain with inter-operator sharing, while $U_z^0$ is the revenue obtained without inter-operator sharing. The details of determining $U_z^0$ will be illustrated in Section \ref{subsec:Uz0}.
% which can be obtained by solving the WSRM in the XXXXXXXXXX.

In the above formulation, we do not take $x_{i,k}$ into consideration in the left-hand-side (LHS) of constraint \eqref{eq:p1_bkhaul} because $R_{i,k}$ will automatically be $0$ when $x_{i,k}=0$ at the optimality. To see this, consider the optimal beamforming and user association scheme of \textbf{P1} for user $k$ is $x_{i,k}=0$ while $\mathbf{w}_{i,k}^n\neq0$. Since $x_{i,k}$ is considered in the objective \eqref{eq:p1_obj}, by setting $\mathbf{w}_{i,k}^n=0$, a feasible solution with higher objective value can be obtained. Then, according to \eqref{eq:sha_sinr} and \eqref{eq:sha_rate}, $R_{i,k}$ will be $0$, too.

\section{Centralized Optimization Algorithms for WSRM}\label{sec:cen_alg}

In this section, we leverage the block successive maximization (BSM) method \cite{razaviyayn2013unified} to solve \textbf{P1}, where the above problems are solved in three stages iteratively. 
In the first stage, the beamforming $\mathbf{w}$ of BSs/STs and power $\mathbf{p}$ of the multi-beam LEO satellite are optimized for the given user association scheme $\mathbf{x}$ and time allocation $\mathbf{t}$ using the SCA method. 
Then, $\mathbf{w}$ and $\mathbf{t}$ are optimized with fixed $\mathbf{x}$ and $\mathbf{p}$.
In the last stage, $\mathbf{x}$ is optimized for the given beamforming design and resource allocation. 
Next, we also provide a method to obtain the revenue for each operator without inter-operator sharing ($U_z^0$).
Besides, we propose an initial point searching method for the SCA-based algorithm.
% as the feasible initial points are crucial for the SCA-based algorithm, 
% we propose an finding feasible initial point algorithm based on the penalty method. to search for the feasible initial points. 
Finally, the proposed centralized algorithm can be implemented in the assigned control center.

% Furthermore, we assume that each operator has a control center to collect necessary information (e.g., channel state information (CSI), and user location.) from its all BSs/STs.

% Based on the collected information, one control center can be assigned to optimize the global user association, resource allocation, and beamforming design.

\subsection{Resource Allocation and Beamforming Optimization for WSRM}\label{subsec:cen_ptw}

% It is obvious that \textbf{P2} can be regarded as a special case of \textbf{P1}.
% Therefore, we focus on the optimization of resource allocation and beamforming design in \textbf{P1} in the following. The algorithm proposed for \textbf{P1} can be applied to optimizing $\{\mathbf{p,t,w}\}$ in \textbf{P2}. 
\textbf{P1} is a non-convex and mixed-integer problem, which is hard to tackle. To make the problem tractable, 
% we first relax $x_{mi,nk}$ by letting $0 \leq x_{mi,nk} \leq 1$. Then, 
we introduce some auxiliary variables and rewrite \textbf{P1} equivalently as
\begin{subequations}\label{eq:p11}
  \begin{align}
    \mathbf{P1.1:} & \max_{\mathbf{p,t,w,x,\Gamma},\boldsymbol{\beta,\varphi,\rho}} \sum_{i\in\mathcal{B}}\sum_{k\in\mathcal{U}} x_{i,k}b_{i,k}r_{i,k} \label{eq:p11_obj} \\
    & \text{s.t. }\eqref{eq:p1_ta}, \eqref{eq:p1_sat_power}, \eqref{eq:p1_ua}, \eqref{eq:p1_binary}, \label{eq:p11_p1_cons}\\
    & \quad\ \
    \Gamma_{i,k}^n \leq \frac{\left\vert\mathbf{h}_{i,k}^n \mathbf{w}_{i,k}^n\right\vert^2}{\beta_{i,k}^n}, \label{eq:p11_sinr_low_bound}\\
    & \quad\ \
    \beta_{i,k}^n \geq \tilde{I}_{i,k}^n + I_{i,k}^n + \sigma_t^2, \label{eq:p11_interfer_up_bound}\\
    & \quad\ \
    \varphi_{i,k}^n \geq \frac{\left\vert\mathbf{h}_{i,k}^n \mathbf{w}_{i,k}^n\right\vert^2}{\rho_{i,k}^n}, \label{eq:p11_sinr_up_bound}\\
    & \quad\ \
    \rho_{i,k}^n \leq \tilde{I}_{i,k}^n + I_{i,k}^n + \sigma_t^2, \label{eq:p11_interfer_low_bound}\\
    & \quad\ \
    \sum_{i\in\mathcal{B}}\sum_{k\in\mathcal{U}}x_{i,k}\alpha_{i,k}^zr_{i,k} \geq U_z^0, \label{eq:p11_mbc}\\
    & \quad\ \
    \sum_{n\in\mathcal{O}}\sum_{k\in\mathcal{U}} B_C\log_2(1+\varphi_{i,k}^n) \leq C_{\ell,i}, \label{eq:p11_bkhaul}\\
    % & \quad\ \
    % 0 \leq x_{mi,nk} \leq 1,  \label{eq:p21_c}\\
    & \quad\ \
    r_{i,k}=\sum_{n\in\mathcal{O}}\log_2(1+\Gamma_{i,k}^n), \label{eq:p11_rate}
  \end{align}
\end{subequations}
where $\mathbf{\Gamma}\triangleq\{\Gamma_{i,k}^n\}_{n\in\mathcal{O}, i\in\mathcal{B},k\in\mathcal{U}}$ represents the lower bound of SINR for user $k$ associated with BS/ST $i$ over band $n$; $\boldsymbol{\beta}\triangleq\{\beta_{i,k}^n\}_{n\in\mathcal{O}, i\in\mathcal{B},k\in\mathcal{U}}$ represents the upper bound of total interference and noise received at user $k$; $\boldsymbol{\varphi}\triangleq\{\varphi_{i,k}^n\}_{n\in\mathcal{O}, i\in\mathcal{B},k\in\mathcal{U}}$ represents the upper bound of SINR for user $k$ associated with BS/ST $i$ over band $n$; $\boldsymbol{\rho}\triangleq\{\rho_{i,k}^n\}_{n\in\mathcal{O}, i\in\mathcal{B},k\in\mathcal{U}}$ represents the lower bound of interference and noise received at user $k$.
% Note that we do not take $x_{i,k}$ into consideration in the left-hand-side (LHS) of constraint \eqref{eq:p11_bkhaul}, because $\varphi_{i,k}^n$ and $\mathbf{w}_{i,k}^n$ must be $0$ when $x_{i,k}=0$ at the optimality. To see this, consider the optimal beamforming and association scheme of \textbf{P1.1} for user $k$ is $x_{i,k}=0$ while $\varphi_{i,k}^n\neq0$. Since $x_{i,k}$ is considered in the objective of \textbf{P1.1} \eqref{eq:p11_obj}, by setting $\varphi_{i,k}^n=0$, a feasible solution with higher objective value can be obtained. Then, according to \eqref{eq:p11_sinr_up_bound}, $\mathbf{w}_{i,k}^n$ will be $0$, too.

By introducing $\mathbf{\Gamma}$ and $\boldsymbol{\varphi}$, the MBC \eqref{eq:mbc} and backhaul constraint \eqref{eq:p11_bkhaul} can be satisfied at the optimality, respectively. Besides, it is obvious that the constraints from \eqref{eq:p11_sinr_low_bound} to \eqref{eq:p11_interfer_low_bound}  should hold with equality at the optimality. Thus, the equivalence between \eqref{eq:p1} and \eqref{eq:p11} is guaranteed.
\textbf{P1.1} is still challenging to tackle due to the non-convex constraints \eqref{eq:p11_sinr_low_bound}, \eqref{eq:p11_interfer_low_bound} and \eqref{eq:p11_bkhaul}. To make \textbf{P1.1} more tractable, we propose to exploit the SCA method to deal with these non-convex constraints.

Note that the right-hand-side (RHS) of \eqref{eq:p11_sinr_low_bound} is a joint convex quadratic-over-linear function with respect to $\mathbf{w}_{i,k}^n$ and $\beta_{i,k}^n$. With this fact, we can use the SCA method to approximate the RHS of \eqref{eq:p11_sinr_low_bound} with a lower bound. The RHS of \eqref{eq:p11_sinr_low_bound} can be approximated with the first-order Taylor expression as
\begin{equation}\label{eq:sca_gamma}
    \frac{\left\vert\mathbf{h}_{i,k}^n \mathbf{w}_{i,k}^n\right\vert^2}{\beta_{i,k}^n} \geq g_1(\mathbf{w}_{i,k}^n, \beta_{i,k}^n),
\end{equation}
where $g_1(\mathbf{w}_{i,k}^n, \beta_{i,k}^n)$ is given in \eqref{eq:sca_gamma_1} on top of the next page,
\begin{figure*}[htbp]
  \begin{equation}\label{eq:sca_gamma_1}
    g_1(\mathbf{w}_{i,k}^n, \beta_{i,k}^n) \triangleq \frac{2\text{Re}\left\{\left(\mathbf{w}_{i,k}^n[\tau-1]\right)^H\mathbf{H}_{i,k}^n\mathbf{w}_{i,k}^n\right\}}{\beta_{i,k}^n[\tau-1]} - \frac{\left\vert\mathbf{h}_{i,k}^n \mathbf{w}_{i,k}^n[\tau-1]\right\vert^2}{(\beta_{i,k}^n[\tau-1])^2}\beta_{i,k}^n
  \end{equation}
  \hrulefill
  \begin{equation}\label{eq:sca_rho_1}
    g_2(\mathbf{w}_{j,k^\prime}^n) \triangleq 2\text{Re}\left\{\left(\mathbf{w}_{j,k^\prime}^n[\tau-1]\right)^H\mathbf{H}_{j,k}^n\mathbf{w}_{j,k^\prime}^n\right\} - \left\vert\mathbf{h}_{j,k}^n \mathbf{w}_{j,k^\prime}^n[\tau-1]\right\vert^2
    %\left(\mathbf{w}_{j,k^\prime}^n[\tau-1]\right)^H\mathbf{H}_{j,k}^n\mathbf{w}_{j,k^\prime}^n[\tau-1] \right) + \sigma_t^2
  \end{equation}
  \hrulefill
  \vspace{-1em}
\end{figure*}
$\mathbf{w}_{i,k}^n[\tau-1]$, $\beta_{i,k}^n[\tau-1]$ are the iterative optimization variables obtained in iteration $\tau-1$, and $\mathbf{H}_{i,k}^n\triangleq \left(\mathbf{h}_{i,k}^n\right)^H\mathbf{h}_{i,k}^n$. Note that the RHS of \eqref{eq:p11_interfer_low_bound} is also a convex function with respect to $\mathbf{w}$, and thus its lower bound can be approximated as
\begin{equation}\label{eq:sca_rho}
  \tilde{I}_{i,k}^n + I_{i,k}^n + \sigma_t^2 \geq \sum_{(j,k^\prime)\neq(i,k)}g_2(\mathbf{w}_{j,k^\prime}^n) + \sigma_t^2,
\end{equation}
where $g_2(\mathbf{w}_{j,k^\prime}^n)$ is given in \eqref{eq:sca_rho_1} on top of the next page.
The RHS of \eqref{eq:p11_bkhaul} is a concave function with respect to $\varphi_{i,k}^n$ and its upper bound can be obtained as
\begin{equation}\label{eq:sca_phi}
  \begin{aligned}
    % x_{ik}\log_2\left(1+\varphi_{i,k}^n\right)&\leq\log_2\left(1+\varphi_{i,k}^n\right)\\
    &\log_2\left(1+\varphi_{i,k}^n\right)\leq g_3(\varphi_{i,k}^n)\\
    &\triangleq \log_2\left(1+\varphi_{i,k}^n[\tau-1]\right)+ \frac{\varphi_{i,k}^n-\varphi_{i,k}^n[\tau-1]}{\ln2(1+\varphi_{i,k}^n[\tau-1])}.
  \end{aligned}
\end{equation}
The LHS of \eqref{eq:p11_bkhaul} can be expressed as 
the difference of concave functions, and its lower bound is as follows:
\begin{equation}\label{eq:sca_bkhaul}
  \begin{aligned}
    C_{\ell,i} &\geq C_{\ell,i}(\mathbf{p}) \\
    & \triangleq B_{Ka}t_{\ell,i}\Bigg\{\log_2\left(\sum_{\ell^\prime}p_{\ell^\prime}G_T(\theta_{\ell^\prime,i})G_R\vert f_i\vert^2+\sigma_s^2\right) \\
    & - C_{\ell,i}^1 -  C_{\ell,i}^2\Bigg\},
  \end{aligned}
\end{equation}
where %$C_{\ell,i}(\mathbf{p}[\tau-1])$ is given in \eqref{eq:sca_bkhaul_1} on top of the next page, and 
$C_{\ell,i}^1$ and $C_{\ell,i}^2$ are given by
\begin{equation}
  \begin{aligned}
    C_{\ell,i}^1 & = \log_2\left(\sum_{\ell^\prime\neq \ell}p_{\ell^\prime}[\tau-1]G_T(\theta_{\ell^\prime,i})G_R\vert f_i\vert^2+\sigma_s^2\right), \\
    C_{\ell,i}^2 & = \sum_{\ell^\prime\neq \ell}\frac{G_T(\theta_{\ell^\prime,i})G_R\vert f_i\vert^2}{\ln2(\sum_{\ell^\prime\neq \ell}p_{\ell^\prime}[\tau-1]G_T(\theta_{\ell^\prime,i})G_R\vert f_i\vert^2+\sigma_s^2)} \\
    & \times \left(p_{\ell^\prime}-p_{\ell^\prime}[\tau-1]\right).\nonumber
  \end{aligned}
\end{equation}
% \begin{figure*}[htbp]
%   \begin{equation}\label{eq:sca_bkhaul_1}
%     C_{\ell,i}(\mathbf{p}[\tau-1])  \triangleq B_{Ka}t_{\ell,i}\left\{\log_2\left(\sum_{\ell^\prime}p_{\ell^\prime}G_T(\theta_{\ell^\prime,i})G_R\vert f_i\vert^2+\sigma_s^2\right) 
%      - C_{\ell,i}^1(\mathbf{p}[\tau-1]) -  C_{\ell,i}^2(\mathbf{p}[\tau-1])\right\}
%      %\log_2\left(\sum_{\ell^\prime\neq \ell}p_{\ell^\prime}[\tau-1]G_T(\theta_{\ell^\prime,i})G_R\vert f_i\vert^2+\sigma_s^2\right) - 
%      %\sum_{\ell^\prime\neq \ell}\frac{G_T(\theta_{\ell^\prime,i})G_R\vert f_i\vert^2}{p_{\ell^\prime}[\tau-1]G_T(\theta_{\ell^\prime,i})G_R\vert f_i\vert^2+\sigma_s^2}(p_{\ell^\prime}-p_{\ell^\prime}[\tau-1])\}
%   \end{equation}
%   \hrulefill
% \end{figure*}

With approximations \eqref{eq:sca_gamma}, \eqref{eq:sca_rho}, \eqref{eq:sca_phi} and \eqref{eq:sca_bkhaul}, we can obtain the reformulated optimization problem at iteration $\tau$ of SCA as follows:
\begin{subequations}\label{eq:p12}
  \begin{align}
    \mathbf{P1.2:} & \max_{\mathbf{p,t,\tilde{w}}} \sum_{i\in\mathcal{B}}\sum_{k\in\mathcal{U}} x_{i,k}b_{i,k}r_{i,k} \label{eq:p12_obj} \\
    & \text{s.t. }\eqref{eq:p11_p1_cons}, \eqref{eq:p11_interfer_up_bound}, \eqref{eq:p11_sinr_up_bound}, \eqref{eq:p11_mbc}, \eqref{eq:p11_rate}, \nonumber\\
    & \quad\ \
    \Gamma_{i,k}^n \leq g_1(\mathbf{w}_{i,k}^n, \beta_{i,k}^n), \label{eq:p12_sinr_low_bound}\\
    & \quad\ \
    \rho_{i,k}^n \leq \sum_{(j,k^\prime)\neq(i,k)}g_2(\mathbf{w}_{j,k^\prime}^n)+\sigma_t^2, \label{eq:p12_interfer_low_bound}\\
    & \quad\ \
    \sum_{n\in\mathcal{O}}\sum_{k\in\mathcal{U}} B_Cg_3(\varphi_{i,k}^n) \leq C_{\ell,i}(\mathbf{p}). \label{eq:p12_bkhaul}
  \end{align}
\end{subequations}
For notational simplicity, we denote variables $\mathbf{\Gamma}$, $\boldsymbol{\beta,\varphi}$, $\boldsymbol{\rho}$, and $\mathbf{w}$ as $\mathbf{\tilde{w}}=\{\mathbf{w,\Gamma}, \boldsymbol{\beta,\varphi,\rho}\}$, because those variables all depend on $\mathbf{w}$.
The problem \eqref{eq:p12} is still not convex because variables $\mathbf{p,t}$ and $\boldsymbol{\varphi}$ are coupled due to the constraint \eqref{eq:p12_bkhaul}. It can be observed that \eqref{eq:p12} is convex with respect to either $\{\mathbf{p}, \mathbf{\tilde{w}}\}$ or $\{\mathbf{t}, \mathbf{\tilde{w}}\}$.
Based on this observation, the BSM method is adopted to solve \eqref{eq:p12}. Specifically, $\{\mathbf{p,t}, \mathbf{\tilde{w}}\}$ can be updated in the following manner: given $\{\mathbf{t,x}\}$, optimize $\{\mathbf{p}, \mathbf{\tilde{w}}\}$; and given $\{\mathbf{p,x}\}$, optimize $\{\mathbf{t}, \mathbf{\tilde{w}}\}$. 
% \begin{itemize} % modify to step figure
%   \item Given $\{\mathbf{t,x}\}$, optimize $\{\mathbf{p}, \mathbf{\tilde{w}}\}$.
%   \item Given $\{\mathbf{p,x}\}$, optimize $\{\mathbf{t}, \mathbf{\tilde{w}}\}$.
%   % \item Given $\{\mathbf{p,t}, \mathbf{\tilde{w}}\}$, optimize $\mathbf{x}$.
% \end{itemize}
% Since \eqref{eq:p12} is convex with respect to $\{\mathbf{p}, \mathbf{\tilde{w}}\}$ or $\{\mathbf{t,\mathbf{\tilde{w}}}\}$ with given $\mathbf{x}$, both $\{\mathbf{p}, \mathbf{\tilde{w}}\}$ and $\{\mathbf{t,\mathbf{\tilde{w}}}\}$ can be easily optimized.

\subsection{User Association Optimization for WSRM}\label{subsec:cen_ua}
% In this subsection, we show how to optimize $\mathbf{x}$ with given $\{\mathbf{p,t}, \mathbf{\tilde{w}}\}$. 
When variables $\{\mathbf{p,t}, \mathbf{\tilde{w}}\}$ are fixed, \eqref{eq:p12} is an integer programming problem with respect to $\mathbf{x}$, which is challenging to tackle. Therefore, we first relax the binary variable constraint \eqref{eq:p1_binary} into $0\leq x_{i,k}\leq 1$. Then, to obtain a high-quality solution for the original integer programming problem, a penalty term $\varrho(x_{i,k}^2-x_{i,k})$ is introduced into the objective \eqref{eq:p12_obj} to force $x_{i,k}$ to be binary, where $\varrho$ is a positive penalty parameter. The reformulated problem is written as follows:
\begin{subequations}\label{eq:p13}
  \begin{align}
    \mathbf{P1.3:} & \max_{\mathbf{x}} \sum_{i\in\mathcal{B}}\sum_{k\in\mathcal{U}} x_{i,k}b_{i,k}R_{i,k} + \varrho\sum_{i\in\mathcal{B}}\sum_{k\in\mathcal{U}} (x_{i,k}^2-x_{i,k}) \\
    & \text{s.t. }\eqref{eq:p1_ua}, \eqref{eq:p11_mbc},\eqref{eq:p12_bkhaul},\label{eq:p13_p1}\\
    & \quad\ \
    0 \leq x_{i,k} \leq 1.  \label{eq:p13_relax_x}
  \end{align}
\end{subequations}
Considering the penalty term is a convex function of $x_{i,k}$, we approximate it with the SCA method, which is given by 
\begin{equation}\label{eq:obj_pen_sca}
  \begin{aligned}
    (x_{i,k}^2-x_{i,k}) & \geq f(x_{i,k}) \\
    & \triangleq (2x_{i,k}[\tau-1]-1)x_{i,k}\\
    & - x^2_{i,k}[\tau-1],
  \end{aligned}
\end{equation}
where $x_{i,k}[\tau-1]$ is the iterative optimization variable obtained in iteration $\tau-1$.
With \eqref{eq:obj_pen_sca}, we can write the problem at iteration $\tau$ of the SCA as
\begin{equation}\label{eq:p14}
  \begin{aligned}
    \mathbf{P1.4:} & \max_{\mathbf{x}} \sum_{i\in\mathcal{B}}\sum_{k\in\mathcal{U}} x_{i,k}b_{i,k}R_{i,k} + \varrho\sum_{i\in\mathcal{B}}\sum_{k\in\mathcal{U}} f(x_{i,k}) \\
    & \text{s.t. }\eqref{eq:p13_p1}, \eqref{eq:p13_relax_x}. 
  \end{aligned}
\end{equation}
\textbf{P1.4} is a linear programming problem with respect to $\mathbf{x}$ and can be easily solved. 
% Similar to \cite{sanjabi2014optimal}, the gradient projection method is exploited to update $\mathbf{x}$. The reason is that if we solve $\mathbf{x}$ directly to the global optimality, $\mathbf{x}$ may fix at $0$ or $1$ and not change during the algorithm. The gradient projection method is given by
% \begin{equation}\label{eq:gradient_projection_x}
%     (x_{i,k}) = P_{\Omega_{x}}\left(x_{i,k}+\eta b_{i,k}r_{i,k}\right),
% \end{equation}
% where $\eta$ is the step length, and $P_{\Omega_{x}}(\cdot)$ is the projection to the set $\Omega_{x}=\{x_{i,k}\mid\eqref{eq:p1_ua}, \eqref{eq:p11_mbc}, \eqref{eq:p12_bkhaul}, \eqref{eq:p13_relax_x}\}$. 
Besides, it can be verified that, when $\varrho\rightarrow\infty$, the solution $\mathbf{x}$ of \textbf{P1.4} will satisfy the binary constraint \eqref{eq:p1_binary} \cite{khamidehi2016joint}. However, if $\varrho$ is too large, the objective of \eqref{eq:p14} will be dominated by the penalty term, and the original objective (i.e., WSR) will be diminished. To avoid this, we first initialize $\varrho$ to a small value to find a good starting point and then gradually increase $\varrho$ as: $\varrho = q\varrho$, where $q>1$.
The overall algorithm is summarized in Algorithm \ref{Alg1}.
\begin{algorithm}[hthp]
  \caption{Centralized Optimization Algorithm for WSRM\label{Alg1}}
  {
  \begin{algorithmic}[1]
    \STATE Initialize $\mathbf{p}[0], \mathbf{t}[0], \mathbf{\tilde{w}}[0], \mathbf{x}[0]$ to feasible values.
    \REPEAT
    \REPEAT
    \STATE Update $\mathbf{p}[\tau], \mathbf{\tilde{w}}[\tau]$ by solving \eqref{eq:p12} with given $\mathbf{t}[\tau-1]$ and $\mathbf{x}[\tau-1]$.
    \STATE Update $\mathbf{t}[\tau], \mathbf{\tilde{w}}[\tau]$ by solving \eqref{eq:p12} with given $\mathbf{p}[\tau]$ and $\mathbf{x}[\tau-1]$.
    \STATE Update $\mathbf{x}[\tau]$ by solving \eqref{eq:p14} with given $\mathbf{p}[\tau], \mathbf{t}[\tau]$ and $\mathbf{\tilde{w}}[\tau]$.
    \UNTIL The objective converges or the maximum number of iterations is reached.
    \STATE $\varrho = q\varrho$.
    \UNTIL $f(x_{i,k})$ is below a certain threshold $\epsilon$ or the maximum number of iterations is reached.
  \end{algorithmic}}
\end{algorithm}
% \vspace{-1.5em}

\subsection{Obtaining Revenue without Inter-operator Sharing}\label{subsec:Uz0}
% Because $U_z^0$ is the revenue that operator $z$ obtains without inter-operator sharing, 
$U_z^0$ can be determined by solving the WSRM problem without inter-operator sharing, which can be formulated as follows
\begin{subequations}\label{eq:p2}
  \begin{align}
    \mathbf{P2:} & \max_{\mathbf{p,t,W,x}} \sum_{i\in\mathcal{B}_z}\sum_{k\in\mathcal{U}_z} x_{i,k}b_{i,k}R_{i,k} \label{eq:p2_obj} \\ \nonumber
    & \text{s.t. }
    \eqref{eq:p1_ta}, \eqref{eq:p1_bkhaul}, \eqref{eq:p1_sat_power}, \eqref{eq:p1_ua},  \\%\sum_{m,i} x_{mi,nk} \leq 1, \label{eq:p1_ua}
    & \quad\ \
    \sum_{k\in\mathcal{U}_z} \parallel\mathbf{w}_{i,k}^z\parallel^2 \leq P_i, \forall i \in \mathcal{B}_z, \label{eq:p2_st_power} \\
    & \quad\ \
    x_{i,k} = \begin{cases}
               0 \text{ or } 1, & \mbox{if } i\in\mathcal{B}_z, k\in\mathcal{U}_z, \\
               0, & \mbox{otherwise},
             \end{cases} \label{eq:p2_x}
    % & \quad\ \
  \end{align}
\end{subequations}
where \eqref{eq:p2_st_power} indicates there is no spectrum sharing, and \eqref{eq:p2_x} indicates that each user can only be associated with the BSs/STs of the same operator. Further, $R_{i,k}$ in \eqref{eq:p2_obj} is given by
\begin{equation}\label{eq:indv_rate}
  R_{i,k} = \log_2(1 + \gamma_{i,k}^{n}),
\end{equation}
where $\gamma_{i,k}^n$ can be obtained by \eqref{eq:sha_sinr} with $I_{i,k}^n=0$.
% By solving \textbf{P2}, we can obtain $U_z^0$.%, which is the revenue of operator $z$ can obtain without spectrum and service sharing.
\textbf{P2} can be decomposed into independent subproblems for different operators. 
Similar to \textbf{P1}, the subproblem for operator $n$ can be equivalently reformulated as
\begin{equation}\label{eq:p21}
  \begin{aligned}
    \mathbf{P2.1:} & \max_{\mathbf{p,t,\tilde{w},x}} \sum_{i\in\mathcal{B}_z}\sum_{k\in\mathcal{U}_z} x_{i,k}b_{i,k}r_{i,k} \\
    & \text{s.t. }\eqref{eq:p1_ta}, \eqref{eq:p1_sat_power}, \eqref{eq:p1_ua}, \eqref{eq:p2_st_power},\eqref{eq:p2_x},\\
    & \quad\ \
    \eqref{eq:p12_sinr_low_bound}-\eqref{eq:p12_bkhaul}.
    % \sum_{m,i} x_{mi,mk} \leq 1, \label{eq:p11_b}\\
    % & \quad\ \
    % \sum_{m,k} \parallel\mathbf{w}_{mi,mk}^m\parallel^2 \leq P_t, \label{eq:p11_d} \\
    % & \quad\ \
    % \Gamma_{mi,mk}^m \leq g(\mathbf{w}_{mi,mk}^m[\tau-1], \beta_{mi,mk}^m[\tau-1]), \label{eq:relax_sinr}\\
    % & \quad\ \
    % \beta_{mi,mk}^m \geq \tilde{I}_{mi,mk}^m + \sigma^2, \label{relax_interfer}\\
    % & \quad\ \
    % r_{mi,mk}=\log_2(1+\Gamma_{mi,mk}^m), \label{eq:p12_d}\\
    % & \quad\ \
    % 0\leq x_{mi,mk}\leq1.  \label{eq:p12_c}
  \end{aligned}
\end{equation}
It is clear that \textbf{P2.1} can be regarded as a special case of \textbf{P1.2}.
Therefore,the algorithm proposed for \textbf{P1.2} can be applied to solving \textbf{P2.1}. 
Different from \textbf{P1.2}, MBCs are not considered in \textbf{P2.1}, and there exists optimal $\mathbf{x}$ of \textbf{P2.1} are binary variables. Hence, we do not need to introduce the penalty term \eqref{eq:obj_pen_sca} to the objective of \eqref{eq:p21}. 
% Note that with given $\{\mathbf{p,t,\tilde{w}}\}$, \textbf{P2.1} is linear in $\mathbf{x}$.
If we solve $\mathbf{x}$ directly to the global optimality, $\mathbf{x}$ will not change during the algorithm and fix at $0$ or $1$. To find a good $\mathbf{x}$, we exploit the gradient projection \cite{sanjabi2014optimal} for updating $\mathbf{x}$, which is given by
%with given  $\{\mathbf{w,\Gamma}, \boldsymbol{\beta}\}$, \textbf{P2.2} is linear in $\mathbf{x}$, and hence there exists a solution on the vertices of the constraint set. 
% Besides, the gradient projection can be exploited for updating $\mathbf{x}$, which is given by
\begin{equation}\label{eq:gradient_projection_x1}
  x_{i,k} = P_{\Omega_x}\left(x_{i,k}+\eta b_{i,k}R_{i,k}\right),
\end{equation}
where 
$P_{\Omega_x}(\cdot)$ is the projection to the set $\Omega_x=\{x_{i,k}\mid\eqref{eq:p1_ua},\eqref{eq:p2_x}\}$, and $\eta$ is the step size. 
The overall algorithm is summarized in Algorithm \ref{Alg3}.
\begin{algorithm}[hthp]
  \caption{Algorithm for Obtaining Revenue without Inter-operator Sharing\label{Alg3}}
  {
  \begin{algorithmic}[1]
  \STATE Initialize $\mathbf{p}[0], \mathbf{t}[0], \mathbf{\tilde{w}}[0], \mathbf{x}[0]$ to feasible values.
  \REPEAT
  \STATE Update $\mathbf{p}[\tau], \mathbf{\tilde{w}}[\tau]$ by solving \eqref{eq:p21} with given $\mathbf{t}[\tau-1]$ and $\mathbf{x}[\tau-1]$.
    \STATE Update $\mathbf{t}[\tau], \mathbf{\tilde{w}}[\tau]$ by solving \eqref{eq:p21} with given $\mathbf{p}[\tau]$ and $\mathbf{x}[\tau-1]$.
  \STATE Update $\mathbf{x}$ with the gradient projection method \eqref{eq:gradient_projection_x1} with given $\mathbf{p}[\tau], \mathbf{t}[\tau]$ and $\mathbf{\tilde{w}}[\tau]$.
  \UNTIL The objective converges or the maximum number of iterations is reached.
  \end{algorithmic}}
\end{algorithm}
% \vspace{-1.5em}

\subsection{Finding Initial Points}
The feasible initial points are important for the SCA-based algorithm. However, due to the backhaul constraint \eqref{eq:p12_bkhaul} and MBC \eqref{eq:p11_mbc}, it is challenging to find feasible initial points for \textbf{P1.2}. To address this issue, we propose a penalty method to find the feasible initial points \cite{tervo2018energy}. 
Specifically, we can obtain initial points by solving the following optimization problem
\begin{subequations}\label{eq:p3}
  \begin{align}
    \mathbf{P3:} & \max_{\mathbf{p,t,\tilde{w},x,s}}%\max_{\substack{\mathbf{p,t,w,x,\Gamma},\\\boldsymbol{\beta,\varphi,\rho},\mathbf{s}}} 
    \sum_{i\in\mathcal{B}}\sum_{k\in\mathcal{U}} x_{i,k}b_{i,k}r_{i,k} - \xi (\sum_{z\in\mathcal{O}} s_{u_z} + \sum_{i\in\mathcal{B}_S} s_{b_i}) \label{eq:p3_obj} \\
    & \text{s.t. }\eqref{eq:p11_p1_cons}, \eqref{eq:p11_interfer_up_bound}, \eqref{eq:p11_sinr_up_bound}, \eqref{eq:p11_rate}, \eqref{eq:p12_sinr_low_bound}\nonumber\\
    & \quad\ \
    \sum_{i\in\mathcal{B}}\sum_{k\in\mathcal{U}} x_{i,k}\alpha_{i,k}^z r_{i,k} + s_{u_z} \geq U_z^0, \label{eq:p3_mbc}\\
    & \quad\ \
    \sum_{n\in\mathcal{O}}\sum_{k\in\mathcal{U}} B_Cg_3(\varphi_{i,k}^n) \leq C_{\ell,i}(\mathbf{p}) + s_{b_i}, \label{eq:p3_bkhaul}\\
    & \quad \ \
    s_{u_z}, s_{b_i} \geq 0, \label{eq:p3_pen_posit}
  \end{align}
\end{subequations}
where $\xi$ is a positive penalty parameter and $\mathbf{s}\triangleq\{s_{u_z},s_{b_i}\}_{z\in\mathcal{O},i\in\mathcal{B}_S}$ are slack variables. The basic idea of the above formulation is forcing the penalty terms to be zero by maximizing \eqref{eq:p3_obj}. \textbf{P3} can also be solved with the BSM method. Besides, the gradient projection can be exploited for updating $\mathbf{x}$ and $\mathbf{s}$, which is given by
% Note that with given $\{\mathbf{p,t,\tilde{w}}\}$, \textbf{P3} is linear in $\mathbf{x}$ and $\mathbf{s}$.
% If we solve $\mathbf{x}$ directly to the global optimality, $\mathbf{x}$ will not change during the algorithm and fix at $0$ or $1$. To find a good feasible $\mathbf{x}$, we exploit the gradient projection \cite{sanjabi2014optimal} for updating $\mathbf{x}$ and $\mathbf{s}$, which is given by
\begin{equation}\label{eq:gradient_projection_x2}
    (x_{i,k},s_{u_z}) = P_{\Omega_{xs}}\left(x_{i,k}+\eta b_{i,k}R_{i,k},s_{u_z}-\eta\xi\right),
\end{equation}
where $P_{\Omega_{xs}}(\cdot)$ is the projection to the set $\Omega_{xs}=\{x_{i,k},s_{u_z}\mid\eqref{eq:p1_ua}, \eqref{eq:p11_mbc}, \eqref{eq:p13_relax_x}\}$. 
The overall algorithm is summarized in Algorithm \ref{Alg2}. 
% Specifically, $\mathbf{p}[0]$ can be initialized by allocating power equally to beams, $\mathbf{t}[0]$ can be initialized by allocating time equally to STs within a beam, $\mathbf{w}$ can be initialized by maximum ratio transmission scheme with equal power, and $\mathbf{x}$ can be initialized to $\frac{1}{N_U+N_U^\prime}$;
% It is worth noting that the MBC cannot be satisfied if the solution $\mathbf{x}$ obtained by Algorithm \ref{Alg1} is not binary or the solution $\mathbf{s}$ obtained by Algorithm \ref{Alg2} are positive, which means that the operators may refuse to sharing its spectrum or services to avoid the revenue loss.
\begin{algorithm}[hthp]
  \caption{Initial Points Searching Algorithm\label{Alg2}}
  {
  \begin{algorithmic}[1]
  \STATE Initialize $\mathbf{p}[0], \mathbf{t}[0], \mathbf{\tilde{w}}[0], \mathbf{x}[0]$ to feasible values.
  \REPEAT
  \STATE Update $\mathbf{p}[\tau], \mathbf{\tilde{w}}[\tau], \mathbf{s}$ by solving \eqref{eq:p3} with given $\mathbf{t}[\tau-1]$ and $\mathbf{x}$.
  \STATE Update $\mathbf{t}[\tau], \mathbf{\tilde{w}}[\tau], \mathbf{s}$ by solving \eqref{eq:p3} with given $\mathbf{p}[\tau-1]$ and $\mathbf{x}$.
  \STATE Update $\mathbf{x,s}$ using gradient projection \eqref{eq:gradient_projection_x2} with given $\mathbf{p}[\tau], \mathbf{t}[\tau]$ and $\mathbf{\tilde{w}}[\tau]$.
  \UNTIL $\Vert\mathbf{s}\Vert_{\infty}$ is below a certain threshold $\epsilon$ or the maximum number of iterations is reached.
  \end{algorithmic}}
\end{algorithm}
\vspace{-1em}

\subsection{Convergence and Complexity Analysis}
Algorithms proposed in this section have similar main process. As a result, we mainly analyze the convergence and complexity of Algorithm \ref{Alg1}, and the results can be applied to Algorithm \ref{Alg3} and \ref{Alg2}.
The convergence of Algorithm \ref{Alg1} is provided as follows.
\begin{prop}\label{prop1}
  Algorithm \ref{Alg1} is guaranteed to converge
\end{prop}
\begin{proof}%\renewcommand{\qedsymbol}{}
  Please refer to Appendix \ref{app:proof1}
\end{proof}
% Step $3$ to $5$ are the main updating steps of Algorithm \ref{Alg1}, and each of them solves a convex problem and obtains the optimal solution.
% Let $F(\mathbf{p}[\tau-1],\mathbf{t}[\tau-1], \mathbf{x}[\tau-1], \tilde{\mathbf{w}}[\tau-1])$ denote the original objective, and $F_{\mathbf{p}[\tau-1]}(\mathbf{p}[\tau-1],\mathbf{t}[\tau-1], \mathbf{x}[\tau-1], \tilde{\mathbf{w}}[\tau-1])$ denote the approximated objective with the SCA method at $\mathbf{p}[\tau-1]$. It can be seen that $F(\mathbf{p}[\tau-1],\mathbf{t}[\tau-1], \mathbf{x}[\tau-1], \tilde{\mathbf{w}}[\tau-1]) = F_{\mathbf{p}[\tau-1]}(\mathbf{p}[\tau-1],\mathbf{t}[\tau-1], \mathbf{x}[\tau-1], \tilde{\mathbf{w}}[\tau-1])$, because we adopt the first-order taylor expression. 
% For notational simplicity, we use $F(\tau-1)$ and $F_{\mathbf{p}}(\tau-1)$ 
% Step $3$ to $5$ of Algorithm \ref{Alg1}, and each of them solves a convex problem and obtains the optimal solution.
% $F(\mathbf{p,t}, \mathbf{x}, \tilde{\mathbf{w}}) = \sum_{i\in\mathcal{B}}\sum_{k\in\mathcal{U}}x_{i,k}b_{i,k}R_{i,k}$, and  $F_{SCA}(\mathbf{p,t}, \mathbf{x}, \tilde{\mathbf{w}}) = \sum_{i\in\mathcal{B}}\sum_{k\in\mathcal{U}}x_{i,k}b_{i,k}r_{i,k}$
The main complexity of Algorithm \ref{Alg1} is caused by solving \eqref{eq:p12} and \eqref{eq:p14} in the inner loop. The worst-case complexity of Algorithm \ref{Alg1} depends on the number of variables, which can be upper bound as
\begin{equation}
  \begin{aligned}
    \mathcal{O}\large\{I_{out}^1I_{in}^1\large(& \left(N_L + N_t(N_B+N_B^\prime)(N_U+N_U^\prime)\right)^4 \\
    &+ \left(N_B^\prime + N_t(N_B+N_B^\prime)(N_U+N_U^\prime)\right)^4 \\
    &+ \left((N_B+N_B^\prime)(N_U+N_U^\prime)\right)^4 \large)\large\}\nonumber
  \end{aligned}
\end{equation}
where $I_{out}^1$ and $I_{in}^1$ are the number of outer and inner iterations, respectively.
% $I_{out}^1I_{in}^1\left(\left(N_L + 2N_t(N_B+N_B^\prime) + 8(N_B+N_B^\prime)(N_U+N_U^\prime)\right)^4 + \left(N_B^\prime + 2N_t(N_B+N_B^\prime) + 8(N_B+N_B^\prime)(N_U+N_U^\prime)\right)^4 + \left((N_B+N_B^\prime)(N_U+N_U^\prime)\right)^4 \right)$

\section{Distributed Algorithms for WSRM}\label{sec:dis_alg}
The centralized algorithm proposed in Section \ref{sec:cen_alg} is challenging to implement in the real SAGIN.
% due to practical problems, like the huge burden of processing all information in a global center and the privacy issue. 
In this section, we propose a distributed algorithm for the WSRM problem, which can be implemented in local centers independently. Specifically, we develop a distributed user association algorithm based on the Lagrangian dual decomposition and a distributed algorithm for beamforming design and resource allocation based on the consensus ADMM.

\subsection{Lagrangian Dual Decomposition-based User Association Design}
With fixed $\{\mathbf{p}, \mathbf{t}, \mathbf{\tilde{w}}\}$, \textbf{P1.1} can be reformulated as 
\begin{subequations}\label{eq:p15}
  \begin{align}
    \mathbf{P1.5:} & \max_{\mathbf{x}} \sum_{i\in\mathcal{B}}\sum_{k\in\mathcal{U}} x_{i,k}b_{i,k}R_{i,k} \label{eq:p14_obj} \\
    & \text{s.t. }\eqref{eq:mbc}, \eqref{eq:p1_ua}, \eqref{eq:p1_binary},%\eqref{eq:p1_bkhaul}, 
  \end{align}
\end{subequations}
It is clear that \textbf{P1.5} is non-convex due to the binary variables $\mathbf{x}$. Although the centralized algorithm can tackle this problem efficiently by relaxing binary variables, it is difficult to implement in the real SAGIN due to high complexity. 
% The centralized algorithm proposed in the subsection \ref{subsec:cen_ua} relaxes the binary variables to make the problem more tractable, and a penalty is introduced to obtain high-quality integer results finally.
% Although the centralized algorithm can solve this problem effectively, it is nearly impossible to implement in the real SAGIN due to its high complexity. 
Keeping that in mind, we propose a distributed user association algorithm based on the Lagrangian dual decomposition \cite{ye2013ua,shen2014distributed}. 

First, the Lagrangian function with respect to constraint \eqref{eq:mbc} is 
\begin{equation}\label{eq:lag_fun}
  L(\mathbf{x}, \boldsymbol{\lambda}) = \sum_{i\in\mathcal{B}}\sum_{k\in\mathcal{U}} x_{i,k}b_{i,k}R_{i,k} + \sum_{z\in\mathcal{O}}\lambda_z\left(U_z - U_z^0\right),
  % \begin{aligned}
  %   L(\mathbf{x}, \boldsymbol{\lambda},\boldsymbol{\mu}) &= \sum_{i,k} x_{i,k}b_{i,k}R_{i,k} + \sum_{z\in\mathcal{O}}\lambda_z\left(U_z - U_z^0\right)\\
  %   &+ \sum_{i\in\mathcal{B}_S}\mu_i\left(C_i - \sum_k x_{i,k}R_{i,k}\right),
  % \end{aligned}
\end{equation}
where $\boldsymbol{\lambda}$ are dual variables introduced for \eqref{eq:mbc}. Then the dual function is give as
\begin{equation}
  g(\boldsymbol{\lambda}) = \begin{cases}
    & \max_{\mathbf{x}} L(\mathbf{x}, \boldsymbol{\lambda})\\
    & \text{s.t. }\eqref{eq:p1_ua}, \eqref{eq:p1_binary}
  \end{cases}.
\end{equation}
Based on the dual function, the Lagrange dual problem of \textbf{P1.5} can be written as
\begin{equation}\label{eq:dual_ua}
  \min_{\boldsymbol{\lambda}} g(\boldsymbol{\lambda}).
\end{equation}
Now, the solution to \textbf{P1.5} can be obtained by solving \eqref{eq:dual_ua}. Note that since \textbf{P1.5} is non-convex and discrete in nature, solving \eqref{eq:dual_ua} is not equivalent to solving \textbf{P1.5}, and there may exist a duality gap. Nevertheless, good primal solutions can often be obtained by solving the dual problem \cite{shen2014distributed}.

Compared with \textbf{P1.3}, \eqref{eq:dual_ua} can be decomposed into independent subproblems for each user and BS/ST and then solved in a distributed manner. Specifically, \eqref{eq:dual_ua} can be solved in two levels, i.e., the inner maximization problem at the user level and the outer minimization problem at the operator level. In the inner layer, $\mathbf{x}$ is optimized to maximize the Lagrangian function \eqref{eq:lag_fun} with given $\boldsymbol{\lambda}$. The maximization of the Lagrangian function \eqref{eq:lag_fun} has the following explicit analytic solution:
\begin{equation}\label{eq:dual_opt_ua}
  x_{ik}^* = \begin{cases} 1, & \text{if } i = \arg\max\limits_j (b_{j,k} + \sum_{z}\lambda_z\alpha_{j,k}^z)R_{j,k}\\
    0, &\text{otherwise}.
  \end{cases}
\end{equation}
We can see that optimal $x_{ik}^*$ can be determined by each user $k$ independently according to \eqref{eq:dual_opt_ua}.
On the other hand, the dual problem \eqref{eq:dual_ua} is convex concerning $\boldsymbol{\lambda}$ and thus can be solved with the subgradient method. With $x_{ik}^*$, the dual variables can be updated as
\begin{equation}\label{eq:dual_var_upd}
  \lambda_z = \left[\lambda_z - \eta\left(U_z - U_z^0\right)\right]^+,
  % \begin{align}
  %   \lambda_z &= \left[\lambda_z - \eta\left(U_z - U_z^0\right)\right]^+,\label{eq:dual_var_upd_mbc} \\
  %   \mu_i &= \left[\mu_i - \eta\left(C_i - \sum_k x_{ik}^x\right)\right]^+,\label{eq:dual_var_upd_bkhaul}
  % \end{align}
\end{equation}
where $x^+=\max\{x,0\}$. Specifically, $\boldsymbol{\lambda}$ can be updated by each operator independently. The above algorithm is compactly written in Algorithm \ref{Alg4}
\vspace{-1em}
\begin{algorithm}[hthp]
  \caption{Distributed Algorithm for User Association\label{Alg4}}
  {
  \begin{algorithmic}[1]
  \STATE Initialize $\mathbf{x}, \boldsymbol{\lambda}$ to feasible values.
  \REPEAT
  % \REPEAT
  \STATE Each user updates $\mathbf{x}$ according to \eqref{eq:dual_opt_ua}.
  % \STATE Each ST updates $\boldsymbol{\mu}$ according to \eqref{eq:dual_var_upd_bkhaul} and broadcasts $\boldsymbol{\mu}$ to all users.
  % \UNTIL The dual objective converges.
  \STATE Each operator updates $\boldsymbol{\lambda}$ according to \eqref{eq:dual_var_upd} and broadcasts $\boldsymbol{\lambda}$ to all users.
  \UNTIL The dual objective converges or the maximum number of iterations is reached.
  \end{algorithmic}}
\end{algorithm}
\vspace{-2em}

\subsection{Consensus ADMM-based Resourca Allocation and Beamforming Design}
% In Section \ref{subsec:cen_ptw}, an algorithm based on the SCA method is proposed to jointly optimize the beamforming of BSs/STs and resource allocation of the LEO satellite. 
Due to the inter-operator interference \eqref{eq:sha_inter_interfer} and the MBC \eqref{eq:mbc}, Algorithm \ref{Alg1} is required to be implemented in a centralized manner, which is challenging to realize in practice. To this end, we propose an algorithm based on the consensus ADMM \cite{boyd2011distributed} to tackle this problem distributedly.  
%\cite{tervo2017distributed}.

First, \textbf{P1.2} can be equivalently rewritten as follows:
\begin{subequations}\label{eq:p121}
  \begin{align}
    % \mathbf{P1.2.1:} & \max_{\substack{\mathbf{p,t,\bar{w}},\\ \boldsymbol{\digamma},\mathbf{\dot{\Gamma}},\boldsymbol{\theta,\dot{\theta},\psi,\dot{\psi}}}} 
    \mathbf{P1.2.1:} & \max_{\substack{\mathbf{p,t,\bar{w}},\\ \mathbf{\dot{\Gamma}},\boldsymbol{\dot{\theta},\dot{\psi}}}} \sum_{z\in\mathcal{O}}\sum_{i\in\mathcal{B}_z}\sum_{k\in\mathcal{U}} x_{i,k}b_{i,k}r_{i,k}^{(z)} \label{eq:p121_obj} \\
    & \text{s.t. }\eqref{eq:p11_p1_cons}, \eqref{eq:p11_sinr_up_bound}, \eqref{eq:p12_bkhaul}, \nonumber\\
    & \quad\ \
    r_{i,k}^{(z)} = \sum_{n\in\mathcal{O}}\log_2(1+\Gamma_{i,k}^{n,(z)}), \label{eq:p121_rate_local}\\
    & \quad\ \
    \sum_{i\in\mathcal{B}_z}\sum_{k\in\mathcal{U}} x_{i,k}\alpha_{i,k}^zr_{i,k}^{(z)} \geq U_z^0, \label{eq:p121_mbc}\\
    & \quad\ \
    \Gamma_{i,k}^{n,(z)}\! \leq\! g_1(\mathbf{w}_{i,k}^n,\! \beta_{i,k}^n), \label{eq:p121_local_gamma}\\
    & \quad\ \
    \beta_{i,k}^n \geq \tilde{I}_{i,k}^n + \sum_{m\in\mathcal{O}\setminus\{z\}}\theta_{m,k}^{n,(z)} + \sigma_t^2,  \label{eq:p121_beta}\\
    & \quad\ \
    \rho_{i,k}^n \leq \sum_{i\in\mathcal{B}_z}\sum_{k^\prime\in\mathcal{U}\setminus\{ k\}}\!g_2(\mathbf{w}_{i,k^\prime}^n) +\! \sum_{m\in\mathcal{O}\setminus\{z\}}\!\psi_{m,k}^{n,(z)}, \label{eq:p121_rho}\\%\forall i\in\mathcal{B}_z
    & \quad\ \
    \theta_{z,k}^{n,(z)} \geq \sum_{i\in\mathcal{B}_z}\sum_{k^\prime\in\mathcal{U}\setminus\{ k\}}\left\vert\mathbf{h}_{i,k}^n\mathbf{w}_{i,k^\prime}^n\right\vert ^2, \label{eq:p121_local_theta}\\%\forall i\in\mathcal{B}\setminus\mathcal{B}_z, 
    & \quad\ \
    \psi_{z,k}^{n,(z)} \leq \sum_{i\in\mathcal{B}_z}\sum_{k^\prime\in\mathcal{U}\setminus\{ k\}}g_2(\mathbf{w}_{i,k^\prime}^n), \label{eq:p121_local_psi}\\
    & \quad\ \
    \Gamma_{i,k}^{n,(z)} = \dot{\Gamma}_{i,k}^n, \label{eq:p121_global_gamma}\\%\forall z \in\mathcal{O},\forall i\in\mathcal{B},\forall k\in\mathcal{U},
    & \quad\ \
    \theta_{z,k}^{n,(m)} = \dot{\theta}_{z,k}^n, \label{eq:p121_global_theta}\\%\forall m \in\mathcal{O},\forall i\in\mathcal{B},\forall k\in\mathcal{U}, 
    & \quad\ \
    \psi_{z,k}^{n,(m)} = \dot{\psi}_{z,k}^n,  \label{eq:p121_global_psi}%\forall m \in\!\mathcal{O},\forall i\in\!\mathcal{B},\forall k\in\mathcal{U},
  \end{align}
\end{subequations}
where $\mathbf{\bar{w}}\triangleq\{\mathbf{\bar{w}}_z\}_{z\in\mathcal{O}}$, and $\mathbf{\bar{w}}_z\triangleq\{\mathbf{w}_{i,k}^n, \beta_{i,k}^n,\varphi_{i,k}^n,\rho_{i,k}^n,\Gamma_{i,k}^{n,(z)},\theta_{i,k}^{n, (z)},\psi_{i,k}^{n, (z)}\}_{n\in\mathcal{O},i\in\mathcal{B}_z,k\in\mathcal{U}}$ represents the local variables of operator $z$. $\Gamma_{i,k}^{n,(z)}$, $\theta_{i,k}^{n, (z)}$ and $\psi_{i,k}^{n, (z)}$ are newly introduced slack variables.
% Similarly, we define $\boldsymbol{\theta}_z$ and $\boldsymbol{\psi}_z$.
% $\boldsymbol{\theta}_z\triangleq\{\theta_{i,k}^{n, (z)}\}$, and $\boldsymbol{\psi}_z\triangleq\{\psi_{i,k}^{n, (z)}\}$ are newly introduced slack variables. 
Specifically, $\Gamma_{i,k}^{n,(z)}$ represents the SINR of user $k$ associated with BS/ST $i$ over band $n$, while $\theta_{i,k}^{n, (z)}$ and $\psi_{i,k}^{n, (z)}$ denote the upper and lower bound for the inter-operator interference that operator $z$ causes to user $k$ over band $n$, respectively. 
The reason we introduce $\theta_{i,k}^{n,(z)}$ is to guarantee MBCs \eqref{eq:p121_mbc}, while that for $\psi_{i,k}^{n,(z)}$ is to satisfy backhaul constraints \eqref{eq:p12_bkhaul}. 

Note that $\Gamma_{i,k}^{n,(z)}$, $\theta_{i,k}^{n,(z)}$ and $\psi_{i,k}^{n,(z)}$ are all obtained and optimized by operator $z$ locally. 
Intuitively, $\Gamma_{i,k}^{n,(z)}$ is a local version of $\Gamma_{i,k}^n$ determined by operator $z$, while $\theta_{i,k}^{n,(z)}$ and $\psi_{i,k}^{n,(z)}$ are the local versions of $I_{j,k}^n,j\in\mathcal{B}\setminus\mathcal{B}_z$. 
On the other hand, $\mathbf{\dot{\Gamma}}\triangleq\{\dot{\Gamma}_{i,k}^n\}$, $\mathbf{\dot{\theta}_{i,k}^n}\triangleq\{\dot{\theta}_{i,k}^n\}$ and $\mathbf{\dot{\psi}_{i,k}^n}\triangleq\{\dot{\psi}_{i,k}^n\}$ are the corresponding global versions of $\Gamma_{i,k}^n$, $\theta_{i,k}^n$ and $\psi_{i,k}^n$, which can be obtained through aggregating local variables across operators. The constraints \eqref{eq:p121_global_gamma}-\eqref{eq:p121_global_psi} can make sure that operators finally reach a consensus on the value of SINR and inter-operator interference, guaranteeing the equivalence of \eqref{eq:p12} and \eqref{eq:p121}. 
% Then, the constraints \eqref{eq:p121_global_gamma}-\eqref{eq:p121_global_psi} can make sure that operators finally reach a consensus on the value of SINR and inter-operator interference, guaranteeing the equivalence of \eqref{eq:p12} and \eqref{eq:p121}. 

With the introduced slack variables, 
% the original constraints 
% \eqref{eq:p11_mbc}, \eqref{eq:p11_interfer_up_bound}, and \eqref{eq:p12_interfer_low_bound} 
% making the operators coupled are transformed into new constraints \eqref{eq:p121_mbc}, \eqref{eq:p121_beta} and \eqref{eq:p121_rho}. Hence, 
the constraints from \eqref{eq:p121_mbc} to \eqref{eq:p121_local_psi} can be decomposed into independent local convex sets for each operator, and the local set for operator $z$ is:
\begin{equation}\label{eq:p13_cons_set}
  \begin{aligned}
    \mathcal{C}_z = \Bigg\{
      & \left(\mathbf{p,t,\bar{w}}_z \right) \mid \eqref{eq:p11_p1_cons}, \eqref{eq:p11_sinr_up_bound}, \eqref{eq:p12_bkhaul}, \eqref{eq:p121_rate_local}-\eqref{eq:p121_local_psi}
      % & \sum_{i,k}x_{i,k}\alpha_{i,k}^zr_{i,k}^{(z)} \geq U_z^0, \\
      % & \Gamma_{i,k}^{n,(z)} \leq g_1(\mathbf{w}_{i,k}^n[\tau-1], \beta_{i,k}^n[\tau-1]), \\
      % & \beta_{i,k}^n \geq \tilde{I}_{i,k}^n + \theta_{m,k}^{n,(z)} + \sigma_t^2, \\
      % & \rho_{i,k}^n \leq g_2(\mathbf{w}_{i,k}^n[\tau-1]) + \psi_{m,k}^{n,(z)}, \\
      % & \theta_{z,k}^{n,(z)} \geq \sum_{i\in\mathcal{B}_z,k^\prime}\left\vert\mathbf{h}_{i,k}^n\mathbf{w}_{i,k^\prime}^n\right\vert ^2,
      % \\
      % & \psi_{z,k}^{n,(z)} \leq \sum_{i\in\mathcal{B}_z,k^\prime}g_2(\mathbf{w}_{i,k^\prime}^n[\tau-1])
    \Bigg\},
  \end{aligned}
\end{equation}
% where $\mathbf{\bar{w}}_z\triangleq\{\mathbf{w}_z,\boldsymbol{\beta}_z,\boldsymbol{\varphi}_z,\boldsymbol{\rho}_z,\boldsymbol{\digamma}_z,\boldsymbol{\theta}_z,\boldsymbol{\psi}_z\}$ represents the local variables of operator $z$. % depends on $\mathbf{w}_z$.
Then, the subproblem of \eqref{eq:p121} for operator $z$ can be compactly expressed as
\begin{subequations}\label{eq:p122}
  \begin{align}
    \mathbf{P1.2.2:} & \min_{\substack{\mathbf{p,t,\bar{w}}_z,\\ \mathbf{\dot{\Gamma}},\boldsymbol{\dot{\theta}},\boldsymbol{\dot{\psi}}}} \sum_{i\in\mathcal{B}_z}\sum_{k\in\mathcal{U}} -x_{i,k}b_{i,k}r_{i,k}^{(z)} \label{eq:p122_obj} \\
    & \text{s.t. } \eqref{eq:p121_rate_local} \nonumber\\
    & \quad\ \
    \left(\mathbf{p,t,\bar{w}}_z\right) \in \mathcal{C}_z,\\
    & \quad\ \
    \Gamma_{i,k}^{n,(z)} = \dot{\Gamma}_{i,k}^n, \forall i\in\mathcal{B},\forall k\in\mathcal{U},\label{eq:p122_gamma}\\
    & \quad\ \
    \theta_{z,k}^{n,(m)} = \dot{\theta}_{z,k}^n, \forall m \in\mathcal{O},\forall i\in\mathcal{B},\forall k\in\mathcal{U},\label{eq:p122_theta}\\
    & \quad\ \
    \psi_{z,k}^{n,(m)} = \dot{\psi}_{z,k}^n, \!\forall m \in\mathcal{O},\!\forall i\in\mathcal{B},\forall k\in\mathcal{U}. \label{eq:p122_psi}
  \end{align}
\end{subequations}
For notational simplicity, define $\boldsymbol{\Phi}_z\triangleq\{\mathbf{p,t,\bar{w}}_z\}$ and $\boldsymbol{\digamma}_z\triangleq\{\Gamma_{i,k}^{n,(z)}\}_{n\in\mathcal{O},i\in\mathcal{B}_z,k\in\mathcal{U}}$. Similarly, we define $\boldsymbol{\theta}_z$ and $\boldsymbol{\psi}_z$. Then the augmented Lagrangian of \eqref{eq:p122} can be written in the following
\begin{subequations}\label{eq:p122_auglag}
  \begin{align*}
    L(\boldsymbol{\Phi}_z,\mathbf{\dot{\Gamma}},\boldsymbol{\dot{\theta},\dot{\psi}},\boldsymbol{\nu}_z)=
    & \sum_{i\in\mathcal{B}_z}\sum_{k\in\mathcal{U}} -x_{i,k}b_{i,k}r_{i,k}^{(z)} + 
    % \sum_{z\in\mathcal{O}}\Bigg\{
      \mathbb{I}_{\mathcal{C}_z}(\boldsymbol{\Phi}_z)\\
    & + \boldsymbol{\nu}_{z, \Gamma}^T\left(\boldsymbol{\digamma}_z-\mathbf{\dot{\Gamma}}\right) + \frac{c}{2}\left\Vert\boldsymbol{\digamma}_z-\mathbf{\dot{\Gamma}}\right\Vert^2\\
    & + \boldsymbol{\nu}_{z, \theta}^T\left(\boldsymbol{\theta}_z-\boldsymbol{\dot{\theta}}\right) + \frac{c}{2}\left\Vert\boldsymbol{\theta}_z-\boldsymbol{\dot{\theta}}\right\Vert^2\\
    & + \boldsymbol{\nu}_{z, \psi}^T\left(\boldsymbol{\psi}_z-\boldsymbol{\dot{\psi}}\right) + \frac{c}{2}\left\Vert\boldsymbol{\psi}_z-\boldsymbol{\dot{\psi}}\right\Vert^2,
    % \Bigg\},
    % & + \boldsymbol{\nu}_{\Gamma}^T\left(\boldsymbol{\digamma}-\mathbf{\dot{\Gamma}}\right)
    % + \boldsymbol{\nu}_{\theta}^T\left(\boldsymbol{\theta}-\boldsymbol{\dot{\theta}}\right) + \boldsymbol{\nu}_{\psi}^T\left(\boldsymbol{\psi}-\boldsymbol{\dot{\psi}}\right)\\
    % & + \frac{c}{2}\left\Vert\boldsymbol{\digamma}-\mathbf{\dot{\Gamma}}\right\Vert^2 + \frac{c}{2}\left\Vert\boldsymbol{\theta}-\boldsymbol{\dot{\theta}}\right\Vert^2,\\
  \end{align*}
\end{subequations}
where $\mathbb{I}_{\mathcal{C}_z}(\boldsymbol{\Phi}_z)$ is an indicator function defined as
\begin{equation}
  \mathbb{I}_{\mathcal{C}_z}(\boldsymbol{\Phi}_z) = \begin{cases}
    0, & \text{if } \boldsymbol{\Phi}_z\in\mathcal{C}_z,\\
    +\infty, & \text{otherwise},
  \end{cases}\nonumber
\end{equation}
$\boldsymbol{\nu}_z\triangleq\{\boldsymbol{\nu}_{z,\Gamma},\boldsymbol{\nu}_{z,\theta},\boldsymbol{\nu}_{z,\psi}\}$ denotes the dual variables related to constraints \eqref{eq:p122_gamma}-\eqref{eq:p122_psi} of operator $z$, and $c$ is the penalty coefficient introduced for penalizing the violation of equality constraints.

In the following, we elaborate on the variable updates in ADMM. First, the local variables $\{\mathbf{p,\bar{w}}_z\}$ or $\{\mathbf{t,\bar{w}}_z\}$ in the $(l+1)$-th iteration can be updated as follows:
\begin{equation}\label{eq:aug_local_upd}
  \begin{aligned}
    (\mathbf{p,\bar{w}}_z)^{(l+1)}
    &= \arg\min_{\mathbf{p,\bar{w}}_z}L(\boldsymbol{\Phi}^{(l)},\mathbf{\dot{\Gamma}}^{(l)},\boldsymbol{\dot{\theta}}^{(l)},\boldsymbol{\dot{\psi}}^{(l)},\boldsymbol{\nu}_z^{(l)}),\\
    (\mathbf{t,\bar{w}}_z)^{(l+1)}
    &= \arg\min_{\mathbf{p,\bar{w}}_z}L(\boldsymbol{\Phi}^{(l)},\mathbf{\dot{\Gamma}}^{(l)},\boldsymbol{\dot{\theta}}^{(l)},\boldsymbol{\dot{\psi}}^{(l)},\boldsymbol{\nu}_z^{(l)}).
  \end{aligned}
\end{equation}
Next, the global variables can be updated as
\begin{equation}\label{eq:aug_glb_upd}
  \begin{aligned}
    \boldsymbol{\dot{\Gamma}}^{(l+1)} &= \frac{1}{\vert\mathcal{O}\vert}\sum_{z\in\mathcal{O}}\left(\boldsymbol{\digamma}^{(l+1)}_z + \frac{1}{c}\boldsymbol{\nu}_{z,\Gamma}^{(l)}\right), \\
    \boldsymbol{\dot{\theta}}^{(l+1)} &= \frac{1}{\vert\mathcal{O}\vert}\sum_{z\in\mathcal{O}}\left(\boldsymbol{\theta}^{(l+1)}_z + \frac{1}{c}\boldsymbol{\nu}_{z,\theta}^{(l)}\right),\\
    \boldsymbol{\dot{\psi}}^{(l+1)} &= \frac{1}{\vert\mathcal{O}\vert}\sum_{z\in\mathcal{O}}\left(\boldsymbol{\psi}^{(l+1)}_z + \frac{1}{c}\boldsymbol{\nu}_{z,\psi}^{(l)}\right),
    % \boldsymbol{\dot{\Gamma}}^{(l+1)} &= \frac{1}{2}(\boldsymbol{\digamma}^{(l)}_G + \boldsymbol{\Gamma}^{S,(l)}) + \frac{1}{2c}(\boldsymbol{\nu}_{\Gamma}^{G,(l)}+\boldsymbol{\nu}_{\Gamma}^{S,(l)}), \\
    % \boldsymbol{\dot{\theta}}^{(l+1)} &= \frac{1}{2}(\boldsymbol{\theta}^{G,(l)} + \boldsymbol{\theta}^{S,(l)}) + \frac{1}{2c}(\boldsymbol{\nu}_{\theta}^{G,(l)}+\boldsymbol{\nu}_{\theta}^{S,(l)}),\\
    % \boldsymbol{\dot{\psi}}^{(l+1)} &= \frac{1}{2}(\boldsymbol{\psi}^{G,(l)} + \boldsymbol{\psi}^{S,(l)}) + \frac{1}{2c}(\boldsymbol{\nu}_{\psi}^{G,(l)}+\boldsymbol{\nu}_{\psi}^{S,(l)}).
  \end{aligned}
\end{equation}
where $\vert\mathcal{O}\vert$ is the number of operators. 
% , while $\boldsymbol{\nu}_{z,\Gamma},\boldsymbol{\nu}_{z,\theta}$, and $\boldsymbol{\nu}_{z,\psi}$ denote the dual variables related to operator $z$
Finally, dual variables can be updated as follows:
\begin{equation}\label{eq:aug_dual_upd}
  \begin{aligned}
    \boldsymbol{\nu}_{z, \Gamma}^{(l+1)} &= \boldsymbol{\nu}_{z, \Gamma}^{(l)} + c(\boldsymbol{\digamma}_z^{(l+1)}-\mathbf{\dot{\Gamma}}^{(l+1)}),\\
    \boldsymbol{\nu}_{z, \theta}^{(l+1)} &= \boldsymbol{\nu}_{z, \theta}^{(l)} + c(\boldsymbol{\theta}_z^{(l+1)}-\boldsymbol{\dot{\theta}}^{(l+1)}),\\
    \boldsymbol{\nu}_{z, \psi}^{(l+1)} &= \boldsymbol{\nu}_{z, \psi}^{(l)} + c(\boldsymbol{\psi}_z^{(l+1)}-\boldsymbol{\dot{\psi}}^{(l+1)}).
  \end{aligned}
\end{equation}
The algorithm for updating $\{\mathbf{p}[\tau], \mathbf{\bar{w}}[\tau]\}$ or $\{\mathbf{t}[\tau], \mathbf{\bar{w}}[\tau]\}$ is summarized in Algorithm \ref{Alg51}.
\begin{algorithm}[hthp]
  \caption{Distributed Algorithm for Resource Allocation and Beamforming\label{Alg51}}
  {
  \begin{algorithmic}[1]
    \STATE Input $\mathbf{p}[\tau-1], \mathbf{t}[\tau-1], \mathbf{\bar{w}}[\tau-1], \mathbf{x}[\tau-1], \boldsymbol{\nu}_z$.
    \REPEAT
    \STATE Each operator updates local variables $\{\mathbf{p}[\tau], \mathbf{\bar{w}}_z[\tau]\}$ or $\{\mathbf{t}[\tau], \mathbf{\bar{w}}_z[\tau]\}$ according to \eqref{eq:aug_local_upd}.
    \STATE Each operator shares $\boldsymbol{\Gamma}_z, \boldsymbol{\theta}_z$ and $\boldsymbol{\psi}_z$ with each other and updates global variables $\boldsymbol{\dot{\Gamma}}, \boldsymbol{\dot{\theta}}$ and $\boldsymbol{\dot{\psi}}$ according to \eqref{eq:aug_glb_upd}.
    \STATE Each operator updates dual variables $\boldsymbol{\nu}_z$ according to \eqref{eq:aug_dual_upd}.
    \UNTIL The objective converges or the maximum number of iterations is reached.
    \IF {The new WSR is smaller than the previous one}
    \STATE Increase the maximum number of ADMM iteration $I_{ADMM}$
    \STATE Return $\{\mathbf{p}[\tau-1], \mathbf{\bar{w}}[\tau-1]\}$ or $\{\mathbf{t}[\tau-1], \mathbf{\bar{w}}[\tau-1]\}$
    \ENDIF
  \end{algorithmic}}
\end{algorithm}
% \textit{Implementation issues:} 

% It is worth noting that the ADMM-based loop usually requires lots of iterations to converge, significantly slowing down the convergence rate. 
Motivated by the fact observed in \cite{nguyen2016distributed}, we adopt the early termination strategy in the ADMM-loop. However, the new objective value may be worse than the previous one due to limited iterations. In this case, we will increase $I_{ADMM}$ and not accept the new results.
% we compare the new and previous objective values after each ADMM loop to decide whether to accept the new optimization results.
% to guarantee the convergence of this algorithm.
The overall algorithm 
% consisting of user association, resource allocation and beamforming design 
is summarized in Algorithm \ref{Alg5}. 
% The primal residual $r^\dag_p$ and dual residual $r^\dag_d$ are defined as
% \begin{equation}
%   \begin{aligned}
%     r^\dag_p &= \left\Vert\boldsymbol{\digamma}-\boldsymbol{\dot{\Gamma}}\right\Vert^2 + \left\Vert\boldsymbol{\theta}-\boldsymbol{\dot{\theta}}\right\Vert^2 + \left\Vert\boldsymbol{\psi}-\boldsymbol{\dot{\psi}}\right\Vert^2,\\
%     r^\dag_d &= 2c^2(\left\Vert\boldsymbol{\dot{\Gamma}}^{(l+1)}-\boldsymbol{\dot{\Gamma}}^{(l)}\right\Vert^2 + \left\Vert\boldsymbol{\dot{\theta}}^{(l+1)}-\boldsymbol{\dot{\theta}}^{(l)}\right\Vert^2\\
%     & + \left\Vert\boldsymbol{\dot{\psi}}^{(l+1)}-\boldsymbol{\dot{\psi}}^{(l)}\right\Vert^2 )
%   \end{aligned}
% \end{equation}

\begin{algorithm}[hthp]
  \caption{Distributed Algorithm for WSRM\label{Alg5}}
  {
  \begin{algorithmic}[1]
    \STATE Initialize $\mathbf{p}[0], \mathbf{t}[0], \mathbf{\bar{w}}[0], \mathbf{x}[0]$ to feasible values.
    \REPEAT
    \STATE Update $\mathbf{p}[\tau], \mathbf{\bar{w}}[\tau]$ via Algorithm \ref{Alg51}.
    % \STATE Update SCA variables $\mathbf{p}[\tau], \mathbf{\tilde{w}}[\tau]$.
    \STATE Update $\mathbf{t}[\tau], \mathbf{\bar{w}}[\tau]$ via Algorithm \ref{Alg51}.
    % \STATE Decide whether accept $\mathbf{t}[\tau], \mathbf{\tilde{w}}[\tau]$.
    % \STATE Update SCA variables $\mathbf{t}[\tau], \mathbf{\tilde{w}}[\tau]$.
    \STATE Update $\mathbf{x}$ via Algorithm \ref{Alg4}.
    \UNTIL The objective converges or the maximum number of iterations is reached.
  \end{algorithmic}}
\end{algorithm}
\vspace{-2em}

\subsection{Convergence and Complexity Analysis}
% Algorithm \ref{Alg5} consists of an outer SCA-based loop and an inner loop. The convergence of SCA has been proved in Proposition \ref{prop1}. We mainly analyze the convergence of the inner loop, including Algorithm \ref{Alg4} and the proposed consensus ADMM algorithm.
The convergence of Algorithm \ref{Alg5} is provided as follows.
\begin{prop}\label{prop2}
  Algorithm \ref{Alg5} is guaranteed to converge.
\end{prop}
\begin{proof}
  Please refer to Appendix \ref{app:prop2}
\end{proof}
The complexity of Algorithm \ref{Alg5} is dominated by Algorithm \ref{Alg4} (line $5$ of Algorithm \ref{Alg5}) and Algorithm \ref{Alg51} (line $3$ and $4$ of Algorithm \ref{Alg5}). The complexity of Algorithm \ref{Alg4} is $\mathcal{O}(I_{UA}(N_B+N_B^\prime)(N_U+N_U^\prime))$, in which $I_{UA}$ is the iteration number of Algorithm \ref{Alg4}. The main complexity of Algorithm \ref{Alg51} is solving \eqref{eq:aug_local_upd}, which is similar to Algorithm \ref{Alg1}. As a result, the complexities for line $3$ and $4$ of Algorithm \ref{Alg5} are $\mathcal{O}(I_{ADMM}(N_L+(N_B+N_B^\prime)(N_U+N_U^\prime))^4)$ and $\mathcal{O}(I_{ADMM}(N_B^\prime+(N_B+N_B^\prime)(N_U+N_U^\prime))^4)$, respectively. 
% First, we analyze the convergence of Algorithm \ref{Alg4}. Because the dual problem \eqref{eq:dual_ua} is a convex problem, the subgradient method is guaranteed to converge to the optimal solution \cite{boyd2003subgradient}. As a result, the convergence of Algorithm \ref{Alg4} is guaranteed. Next, we analyze the convergence of proposed consensus ADMM algorithm. We use the consensus ADMM algorithm to solve convex problems in step $3$ to $7$ (or step $8$ to $12$) of Algorithm \ref{Alg5}, and the objective will converge with enough iterations \cite{boyd2011distributed}. Besides, the proposed consensus ADMM algorithm 
% Since we adopt the early termination strategy, the convergence results for classic consensus ADMM algorithms \cite{boyd2011distributed} cannot be directly used here. However, the proposed consensus ADMM algorithm can generate a non-decreasing objective because we adopt the trick.
% Since we use the consensus ADMM algorithm to solve convex problems in step $3$ to $7$ (or step $8$ to $12$) of Algorithm \ref{Alg5}, the consensus ADMM algorithm will finally converge with enough iterations \cite{boyd2011distributed}. 

%Besides, the problem \eqref{eq:p13} can be decomposed into independent subproblems for each operator.

\section{Simulation Results}\label{sec:sim_result}

In this section, simulation results are provided to demonstrate the effectiveness of the proposed algorithms and the significance of the MBCs. The considered SAGIN is shown in Fig. \ref{fig:env_model}. Specifically, GNO has $N_B=2$ BSs and $N_U=3$ users, while SNO has $N_B^\prime=4$ STs and $N_U^\prime=7$ users. Each BS/ST is equipped with $N_t=4$ antennas, and the maximum transmit power for each BS is $52$ dBm, and for each ST, it is $49$ dBm. Besides, we assume each operator has a $B_C = 100$ MHz dedicated C band. The parameters of the satellite-ground link are summarized in Table \ref{tab:sys_para}.
The channel for the ground-ground link is set as follows: the path-loss is modeled as $32.4 + 20\log(f_C) + 30\log(d)$, where $f_C=3$ is the carrier frequency in GHz and $d$ is the distance in km. The Rayleigh fading is adopted to characterize the small-scale fading, and the noise power spectral density is $-174$ dBm/Hz. Other parameters are set as follows: $\varrho=1\times10^{-4}$, $q=50$, $\epsilon=20$, $\eta=0.1$, $\epsilon=10^{-3}$, $c=1.5$, and $b_{i,k}=1, \forall i,k$.

\begin{figure}[htbp]
  \begin{center}
  \epsfxsize=0.3\textwidth \leavevmode
  \epsffile{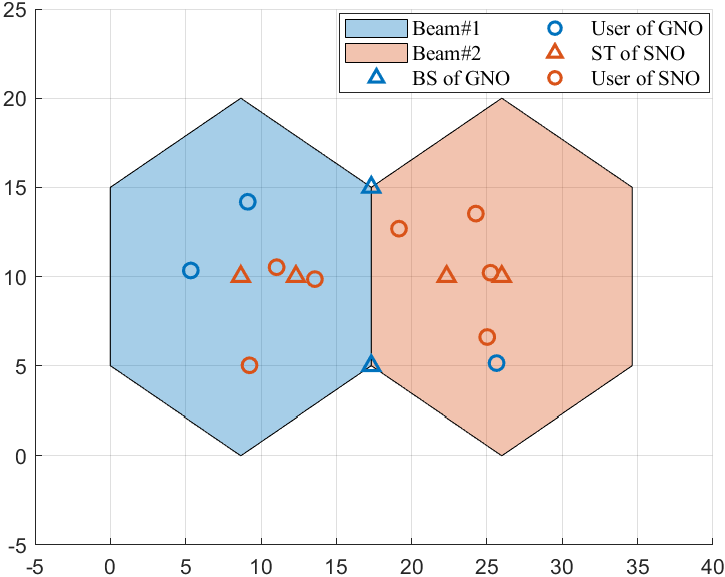}
  \caption{The considered SAGIN.}\label{fig:env_model}
  \end{center}
  \vspace{-2em}
\end{figure}

\begin{table}[htbp]
  \centering
  \caption{System Parameters for Satellite-Ground Link}\label{tab:sys_para}
  \renewcommand\arraystretch{1.3}
  \resizebox{.4\textwidth }{!}{
  \begin{tabular}{c| c}
    \hline
    PARAMETER & VALUE \\ \hline
    \hline 
    LEO height & $600$ km \\ \hline
    The number of beams & $N_L=2$  \\ \hline
    Beam radius & $10$ km \\ \hline
    Carrier frequency & $f_{Ka} = 20$ GHz \\ \hline
    System bandwidth & $B_{Ka}=400$ MHz \\\hline
    LEO maximum transmit power & $P_{Sat}=50$ W \\ \hline
    LEO maximum transmit gain & $G_T^0$ = $40$ dBi \\ \hline
    LEO antenna aperture & $0.6$ m  \\ \hline
    ST maximum receive gain & $G_R=10$ dBi  \\ \hline
    ST antenna temperature & $150$ K  \\ \hline
    Environment temperature & $290$ K  \\ \hline
    Boltzmann constant & $1.38 \times 10 ^ {-23}$ J$/$K  \\ \hline
    Noise figure & $1.2$ dB  \\ \hline
    Path-loss & $92.44 + 20\log(f_{Ka}) + 20\log(d)$ \\ \hline
    SR fading $(b,m,\Omega)$ & $(0.126,10.1,0.835)$ \\ \hline
  \end{tabular}
  }
\end{table}

To show the effectiveness of the proposed algorithms, we consider the following baseline algorithms for comparison.
\begin{itemize}
  \item \textbf{Closest Association and Equal Resource Allocation (CA-ERA):} Each user is connected to the closest BS/ST, with equal power allocation and maximum ratio transmit (MRT) beamfroming employed by BSs/STs. The LEO satellite equally allocates power to beams and time to STs within a beam.
  \item \textbf{Closest Association and Optimized Power and Beamforming (CA-OPW):} The power allocation of the LEO satellite and the beamforming of BSs/STs are optimized with the proposed algorithm.
  \item \textbf{Closest Association and Optimized Time and Beamforming (CA-OTW):} The time allocation of each beam and the beamforming of BSs/STs are optimized with the proposed algorithm.
\end{itemize}

\begin{figure}[htbp]
  \vspace{-1em}
  \begin{center}
  \epsfxsize=0.35\textwidth \leavevmode
  \epsffile{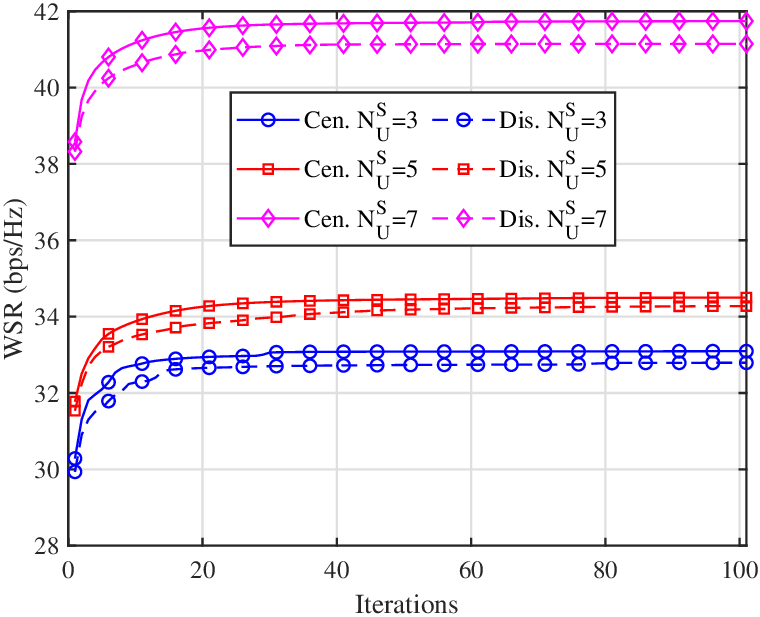}
  \caption{The convergence of the centralized Algorithm \ref{Alg1} and distributed Algorithm \ref{Alg5}: $\delta_G=\delta_S=0.6$.}\label{fig:alg1_converge}
  \end{center}
  \vspace{-1em}
\end{figure}

\begin{figure}[htbp]
  \vspace{-1em}
  \begin{center}
  \epsfxsize=0.35\textwidth \leavevmode
  \epsffile{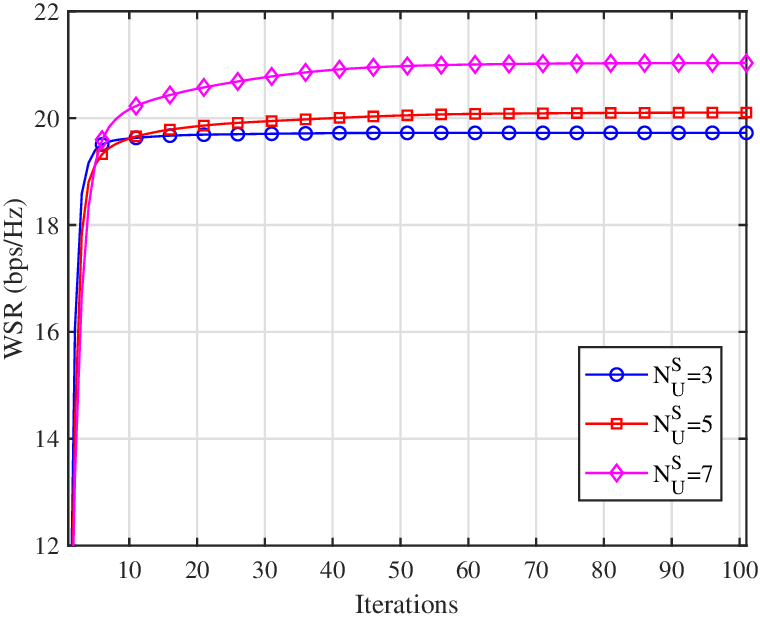}
  \caption{The convergence of Algorithm \ref{Alg3}: $\delta_G=\delta_S=0.6$.}\label{fig:alg3_converge}
  \end{center}
  \vspace{-1.5em}
\end{figure}

\begin{figure}[htbp]
  % \vspace{-1em}
  \begin{center}
  \epsfxsize=0.35\textwidth \leavevmode
  \epsffile{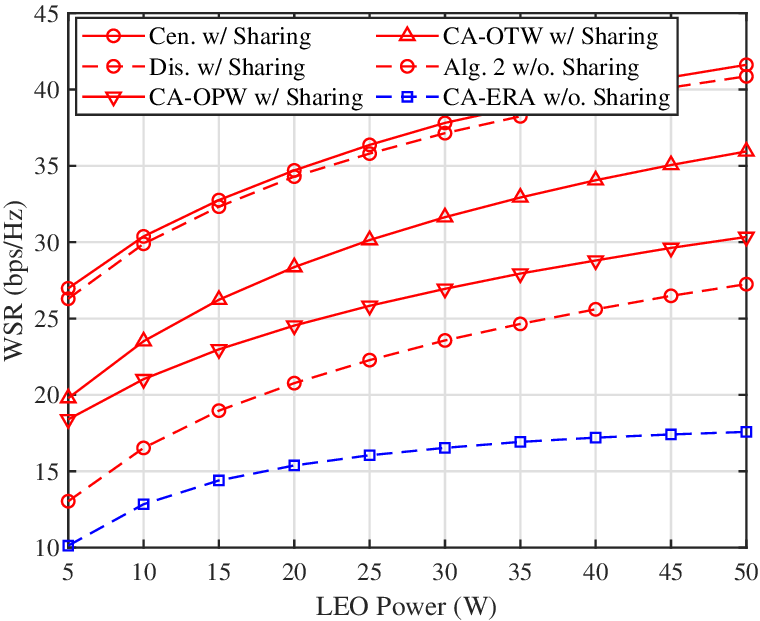}
  \caption{WSR versus maximum transmit power of the LEO satellite $P_{Sat}$: $\delta_G=\delta_S=0.6$, $P_i=52$ dBm,$i\in\mathcal{B}_G$, $P_j=49$ dBm,$j\in\mathcal{B}_S$.}\label{fig:wsr_vs_sat}
  \end{center}
  \vspace{-1em}
\end{figure}

\begin{figure}[htbp]
  % \vspace{-1em}
  \begin{center}
  \epsfxsize=0.35\textwidth \leavevmode
  \epsffile{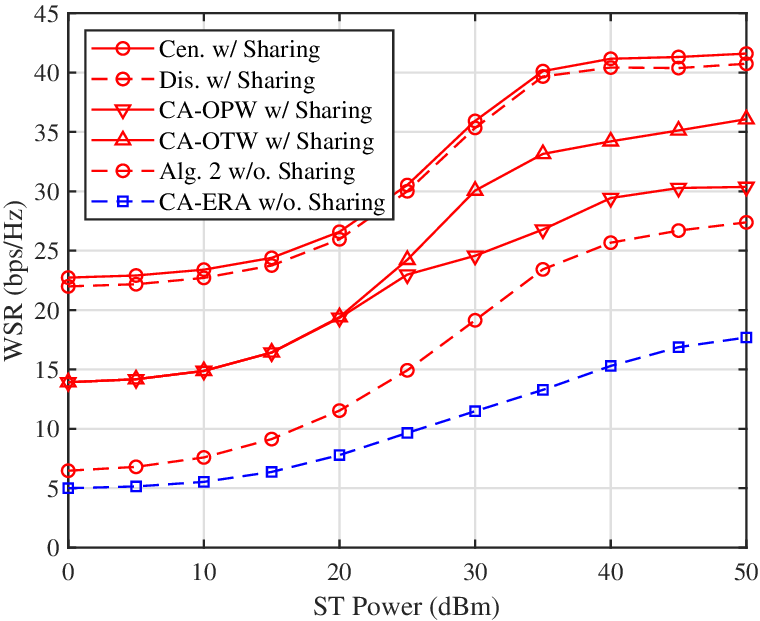}
  \caption{WSR versus maximum transmit power of STs $P_i,i\in\mathcal{B}_S$: $\delta_G=\delta_S=0.6$, $P_{Sat}=50$ W.}\label{fig:wsr_vs_st}
  \end{center}
  \vspace{-1em}
\end{figure}

Firstly, we validate the convergence of proposed Algorithms.
% \ref{Alg1}, Algorithm \ref{Alg3} and Algorithm \ref{Alg5}.
In Fig. \ref{fig:alg1_converge}, we show the convergence of the proposed centralized Algorithm \ref{Alg1} (`Cen.') and distributed Algorithm \ref{Alg5} (`Dis.') when $\delta_G=\delta_S=0.6$. It can be seen that both algorithms converge finally, and the distributed algorithm achieves a very close performance to the centralized algorithm.
Then, the convergence of the Algorithm \ref{Alg3} is demonstrated in Fig. \ref{fig:alg3_converge}. 
% It can be observed that each curve finally converges to a specific value within a limited number of iterations.

Secondly, we compare WSRs achieved by different approaches versus the transmit power of the LEO satellite in Fig. \ref{fig:wsr_vs_sat}. % $P_{Sat}$
It is clear that WSRs increase with $P_{Sat}$ increasing, indicating that the impact of backhaul capacity on the WSR.
Then, we compare WSRs achieved by Algorithm \ref{Alg3} and CA-ERA for the case without inter-operator sharing (`w/o. Sharing'). The proposed algorithm outperforms CR-ERA, demonstrating the effectiveness of Algorithm \ref{Alg3}.
Besides, we illustrate the WSRs for different approaches in the case with inter-operator sharing (`w/ Sharing'). The distributed algorithm approaches the centralized algorithm and performs much better than other benchmarks. This points out that the necessity of jointly optimizing user association, resource allocation  and beamforming design. 
% On the one hand, we compare WSRs achieved by proposed algorithms and those by benchmarks in both cases. It can be seen that 
% We compare WSRs achieved by Algorithm \ref{Alg3} and CA-ERA for the former case. 
% Then, we compare the WSRs obtained by Algorithm \ref{Alg1}, CA-OPW and CA-OTW. The results illustrate that Algorithm \ref{Alg1} performs much better than the other approaches and realizes the highest WSR, which points out that jointly optimizing resource allocation of the LEO satellite and the beamforming design of BSs and STs is necessary. 
Moreover, the WSR achieved in the case with inter-operator sharing is much higher than that in the case without inter-operator sharing, indicating significant gains introduced by realizing SAGIN with inter-operator sharing. 
% In addition, the WSR achieved by the distributed algorithm is close to that achieved by the centralized algorithm.
% Since the two benchmarks cannot guarantee the MBCs, we study the WSRs obtained in the case without inter-operator sharing (`w/o. Sharing'). It can be seen that the proposed algorithm can achieve the highest WSR by jointly optimizing user association and beamforming. Moreover, we also demonstrate the WSRs obtained by the proposed algorithm in the case with inter-operator sharing (`w/ Sharing'). It is evident that inter-operator sharing can lead to significant improvements in WSR.

Thirdly, we investigate the WSRs achieved by different approaches versus the transmit power of STs in Fig. \ref{fig:wsr_vs_st}.% $P_{i}, i\in\mathcal{B}_S$
It can be seen that the WSR in the case with inter-operator sharing does not keep increasing as $P_i$ increases. Particularly, when $P_i$ is greater than $40$ dBm, the WSR almost remains the same. This indicates that the backhaul capacity provided by the LEO satellite limits the WSR, which is consistent with the result obtained from Fig. \ref{fig:wsr_vs_sat}. Additionally, by comparing WSRs achieved by different approaches in two cases, we can find that the proposed algorithms always realize the highest WSRs and the distributed algorithm performs close to that of the centralized algorithm.

\begin{figure}[htbp]
  \vspace{-1em}
  \begin{center}
  \epsfxsize=0.35\textwidth \leavevmode
  \epsffile{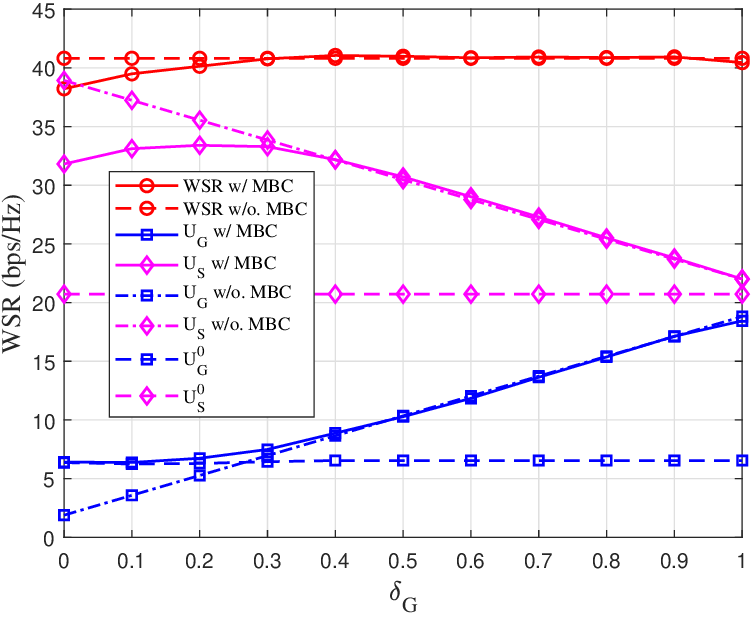}
  \caption{WSR and operators' revenue versus sharing coefficient of GNO $\delta_G$: $\delta_S=1$.}\label{fig:mbc_effect1}
  \end{center}
  \vspace{-1em}
\end{figure}

\begin{figure}[tbp]
  \vspace{-0.5em}
  \begin{center}
  \epsfxsize=0.35\textwidth \leavevmode
  \epsffile{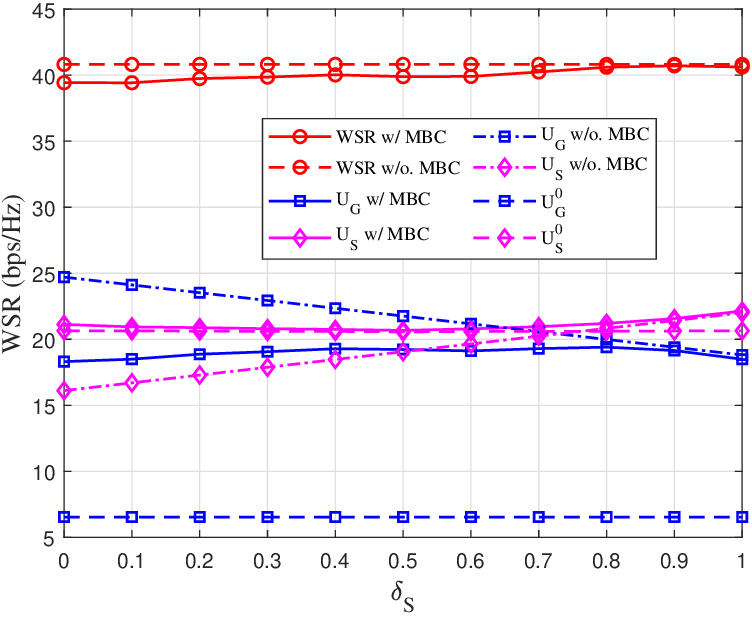}
  \caption{WSR and operators' revenue versus sharing coefficient of SNO $\delta_S$: $\delta_G=1$.}\label{fig:mbc_effect2}
  \end{center}
  \vspace{-1em}
\end{figure}

Next, we present the WSR and operators' revenue versus sharing coefficients to investigate the effect of MBCs in Fig. \ref{fig:mbc_effect1} and Fig. \ref{fig:mbc_effect2}. Specifically, we consider two cases: i) the symbiotic case with MBCs (`w/ MBC'), and ii) the non-symbiotic case without MBCs (`w/o. MBC'). 
In Fig. \ref{fig:mbc_effect1}, we demonstrate WSRs and revenue versus $\delta_G$ with $\delta_S=1$. 
By comparing revenue in the two cases, we can find that each operator's revenue in the symbiotic case does not experience revenue loss after sharing spectrum and services. On the contrary, the GNO suffers revenue loss in the non-symbiotic case when $\delta_G$ is less than $0.2$. This indicates that MBCs can enable operators to achieve mutual benefits.
% This means each operator benefits from inter-operator sharing in the symbiotic case and indicates the importance of the MBC in realizing mutualism. 
Then, we compare the WSRs in the two cases. The WSRs in the symbiotic case are less than those in the non-symbiotic case when $\delta_G<0.2$.
%  and almost coincide with those when $\delta_G\geq0.2$.
The reason is that when $\delta_G<0.2$, the spectrum and service configuration maximizing the WSR makes GNO suffer revenue loss. Therefore, the GNO does not adopt this configuration in the symbiotic case.
% the GNO in the symbiotic case does not adopt the spectrum and service configuration that can maximize the WSR due to the MBC.
It is worth noting that the WSRs achieved in the symbiotic case are close to those in the non-symbiotic case when $\delta_G>0.2$, which demonstrates the proposed algorithms can efficiently maximize WSR and achieve mutual benefits.
% will suffer revenue if only maximizing the overall WSR of the whole considered SAGIN, which can be seen in Fig. \ref{fig:mbc_effect1}. In the symbiotic case, 
% which demonstrates the proposed algorithm can efficiently improve WSR.
Besides, the WSR and revenue of GNO both grow gradually as $\delta_G$ increases, while the revenue of the SNO increases first and then decreases when $\delta_G>0.3$. This trend can be explained as follows. 
% As shown in Fig. \ref{fig:env_model}, some users of SNO are closer to BSs of GNO, while some GNO users are more proximate to STs. Consequently, higher WSR can be realized if service sharing is enabled between SNO and GNO. However,
% the GNO prefers to serve its own users to avoid revenue loss when $\delta_G<0.2$, leading to a low WSR.
With $\delta_G$ increasing, the GNO can obtain more revenue from serving users of SNO. As a result, the GNO demonstrates a growing preference for inter-operator sharing, increasing the WSR and revenue of operators. However when $\delta_G>0.3$,
% it can be found that the revenue of SNO begins to decrease. The reason is that 
the GNO takes away the most revenue from serving SNO's users, resulting in revenue loss of SNO.

Lastly, we study the WSR and revenue versus the sharing coefficient of SNO $\delta_S$ in Fig. \ref{fig:mbc_effect2}. It can be seen that the MBC for SNO cannot be satisfied until $\delta_S\geq0.8$ in the non-symbiotic case. Besides, the WSR and revenue in the symbiotic case almost remain the same when $\delta_S<0.8$. The reason is that the SNO cannot benefit enough from inter-operator sharing and lacks the motivation to share spectrum and services.
On the other hand, by comparing Fig. \ref{fig:mbc_effect1} and Fig. \ref{fig:mbc_effect2}, it can be observed that the GNO will obtain benefit when $(\delta_G\geq0.3, \delta_S=1)$ while the SNO gains when $(\delta_S\geq0.8, \delta_G=1)$. This result indicates a principle to achieve mutual benefits when inter-operator sharing is enabled: the operator with more resources, like the SNO in this paper, should get more compensation from others.
% to guarantee its revenue and realize the mutual benefits.
On the contrary, the GNO, who contributes less to the SAGIN, should not require much compensation from others.

\section{Conclusion}\label{sec:conclu}

In this paper, we have studied a SAGIN consisting of a GNO and SNO, and inter-operator sharing has been adopted to construct such a large-scale network. 
% which can improve resource utilization efficiency and reduce network construction costs.
To fully leverage the advantage of inter-operator sharing, both system performance and individual benefits should be considered.
% Considering that achieving mutual benefits between network operators can facilitate inter-operator sharing, 
Therefore, we have investigated the spectrum and service sharing in the SAGIN from a symbiotic communication perspective, which 
% is an emerging communication paradigm that
can instruct different operators to achieve mutual benefits. Specifically, we have studied a WSR maximization problem about jointly optimizing the user association, resource allocation, and beamforming design. Besides, we have introduced a sharing coefficient to characterize the revenue of each operator, based on which the MBC has been formulated to guarantee that revenue for all operators does not decline after inter-operator sharing. Then, we have proposed a centralized algorithm based on SCA to solve the WSR maximization problem. Since implementing the centralized is difficult in real networks, we have also developed a distributed algorithm based on Lagrangian dual decomposition and consensus ADMM. Finally, extensive simulation results have been provided to show the effectiveness of the proposed algorithms and that the distributed algorithm can approach the centralized algorithm.
Moreover, the results have indicated that the MBCs can enable the operators to realize symbiotic communication in the SAGIN.

\appendices
\section{PROOF OF PROPOSITION \ref{prop1}}\label{app:proof1}
Let $F(\mathbf{p}[\tau-1],\mathbf{t}[\tau-1], \mathbf{x}[\tau-1], \tilde{\mathbf{w}}[\tau-1])$ denote the original objective, and $F_{\mathbf{p}[\tau-1]}(\mathbf{p}[\tau-1],\mathbf{t}[\tau-1], \mathbf{x}[\tau-1], \tilde{\mathbf{w}}[\tau-1])$ denote the approximated objective with the SCA method at point $\mathbf{p}[\tau-1]$. For notational simplicity, we use $F(\mathbf{p}[\tau-1],\tilde{\mathbf{w}}[\tau-1])$ and $F_{\mathbf{p}[\tau-1]}(\mathbf{p}[\tau-1],\tilde{\mathbf{w}}[\tau-1])$ in the following. In step $4$ of Algorithm \ref{Alg1}, we have the following inequality
  \begin{equation}\label{eq:ineq_alg1}
    \begin{aligned}
      F(\mathbf{p}[\tau-1],\tilde{\mathbf{w}}[\tau-1]) &\overset{(a)}{=} F_{\mathbf{p}[\tau-1]}(\mathbf{p}[\tau-1],\tilde{\mathbf{w}}[\tau-1]) \\
      &\overset{(b)}{\leq} F_{\mathbf{p}[\tau-1]}(\mathbf{p}[\tau],\tilde{\mathbf{w}}[\tau]) \\
      &\overset{(c)}{\leq} F(\mathbf{p}[\tau],\tilde{\mathbf{w}}[\tau])
    \end{aligned}
  \end{equation}
where $(a)$ comes from the fact that the first-order Taylor expressions used in \eqref{eq:sca_gamma}, \eqref{eq:sca_rho}, and \eqref{eq:sca_bkhaul} are tight at point $\mathbf{p}[\tau-1]$; $(b)$ holds since $(\mathbf{p}[\tau],\tilde{\mathbf{w}}[\tau])$ are the optimal solutions to \eqref{eq:p12} with given $(\mathbf{t}[\tau-1],\mathbf{x}[\tau-1])$; and $(c)$ holds since $F_{\mathbf{p}[\tau-1]}(\mathbf{p}[\tau],\tilde{\mathbf{w}}[\tau])$ is a lower bound of $F(\mathbf{p}[\tau],\tilde{\mathbf{w}}[\tau])$.

Note that $F_{\mathbf{p}[\tau]}(\mathbf{p}[\tau],\tilde{\mathbf{w}}[\tau])$ is the input of step $5$, and the inequality \eqref{eq:ineq_alg1} can be applied to analyzing step $5$ and $6$. Finally, we have $F(\mathbf{p}[\tau-1],\mathbf{t}[\tau-1], \mathbf{x}[\tau-1], \tilde{\mathbf{w}}[\tau-1]) \leq F(\mathbf{p}[\tau],\mathbf{t}[\tau], \mathbf{x}[\tau], \tilde{\mathbf{w}}[\tau])$, which implies that Algorithm \ref{Alg1} can generate a non-decreasing objective after each iteration. Moreover, the objective is constrained due to the power budget of BSs/STs and the satellite. As a result, Algorithm \ref{Alg1} is guaranteed to converge.

\section{PROOF OF PROPOSITION \ref{prop2}}\label{app:prop2}
In step $3$ of Algorithm \ref{Alg5}, we have the following inequality
  \begin{equation}\label{eq:ineq_alg5_admm}
    \begin{aligned}
      F_{\mathbf{p}[\tau-1]}(\mathbf{p}[\tau-1],\tilde{\mathbf{w}}[\tau-1]) & \overset{(a)}{\leq} F_{\mathbf{p}[\tau-1]}(\mathbf{p}[\tau],\tilde{\mathbf{w}}[\tau])
    \end{aligned}
  \end{equation}
where $(a)$ comes from the fact that the subproblem \eqref{eq:p122} is convex with fixed $(\mathbf{t,x})$ and consensus ADMM can converge to the optimal solution with enough iterations \cite{boyd2011distributed}. Besides, it is clear that $(a)$ also holds when early termination strategy is adopted due to the Algorithm \ref{Alg51} (step $7$). The inequality \eqref{eq:ineq_alg5_admm} is also applicable for step $4$ of Algorithm \ref{Alg5}. As a result, we have $F(\mathbf{p}[\tau-1],\mathbf{t}[\tau-1], \mathbf{x}[\tau-1], \tilde{\mathbf{w}}[\tau-1]) \leq F(\mathbf{p}[\tau],\mathbf{t}[\tau], \mathbf{x}[\tau-1], \tilde{\mathbf{w}}[\tau])$.

In step $6$ of Algorithm \ref{Alg5}, we have the following inequality
\begin{equation}\label{eq:ineq_alg5_alg4}
  \begin{aligned}
    F_{\mathbf{x}[\tau-1]}(\mathbf{x}[\tau-1]) & \overset{(b)}{\leq} F_{\mathbf{x}[\tau-1]}(\mathbf{x}[\tau])\leq F(\mathbf{x}[\tau])
  \end{aligned}
\end{equation}
where $(b)$ holds since the subgradient method is guaranteed to converge to the optimal solution when the dual problem \eqref{eq:dual_ua} is convex \cite{boyd2003subgradient}. As a result, we have $F(\mathbf{p}[\tau-1],\mathbf{t}[\tau-1], \mathbf{x}[\tau-1], \tilde{\mathbf{w}}[\tau-1]) \leq F(\mathbf{p}[\tau],\mathbf{t}[\tau], \mathbf{x}[\tau], \tilde{\mathbf{w}}[\tau])$, which implies that Algorithm \ref{Alg5} generates a non-decreasing objective after each iteration. Therefore, Algorithm \ref{Alg5} is guaranteed to converge. 

% \section*{Acknowledgement}
% This work was supported in part by the National Natural Science Foundation of China under Grant U1801261, the Science and Technology Development Fund, Macau SAR, under Grant 0009/2020/A1, the National Key R\&D Program of China under Grant 2018YFB1801105, the Key Areas of Research and
% Development Program of Guangdong Province, China, under Grant 2018B010114001, the Fundamental Research Funds for the Central Universities under Grant ZYGX2019Z022, and the Programme of Introducing Talents of Discipline to Universities under Grant B20064.

\bibliographystyle{IEEEtran}
\bibliography{IEEEabrv, ref_sagin_share}

\end{document}